\documentclass[12pt]{article}

\usepackage{latexsym,amsmath,amssymb,amsthm,amsfonts,graphicx,natbib,breqn,epsfig,bm,authblk,url}

\usepackage{float}
\usepackage{booktabs} % Added by MQ
\usepackage[title]{appendix} % Added by MQ
\usepackage{graphicx} % Added by MQ
\usepackage[margin=1cm]{caption} % Added by MQ
\usepackage{xr} % added by MQ
\usepackage{enumitem} % added by MQ
\usepackage{mathtools} % Added by MQ
\newcommand*{\patchAmsMathEnvironmentForLineno}[1]{%
      \expandafter\let\csname old#1\expandafter\endcsname\csname #1\endcsname
      \expandafter\let\csname oldend#1\expandafter\endcsname\csname end#1\endcsname
      \renewenvironment{#1}%
         {\linenomath\csname old#1\endcsname}%
         {\csname oldend#1\endcsname\endlinenomath}}%
    \newcommand*{\patchBothAmsMathEnvironmentsForLineno}[1]{%
      \patchAmsMathEnvironmentForLineno{#1}%
      \patchAmsMathEnvironmentForLineno{#1*}}%
    \AtBeginDocument{%
    \patchBothAmsMathEnvironmentsForLineno{equation}%
    \patchBothAmsMathEnvironmentsForLineno{align}%
    \patchBothAmsMathEnvironmentsForLineno{flalign}%
    \patchBothAmsMathEnvironmentsForLineno{alignat}%
    \patchBothAmsMathEnvironmentsForLineno{gather}%
    \patchBothAmsMathEnvironmentsForLineno{multline}%
    }
\usepackage[mathlines,displaymath]{lineno}

\floatstyle{plain}
\usepackage[ruled]{algorithm2e}
\usepackage{algpseudocode}%
\newtheorem{lemma}{Lemma} % Added by MQ
\newtheorem{definition}{Definition} % Added by MQ
\newtheorem{remark}{Remark} % Added by MQ

\newcommand{\Eqref}[1]{Eq. \eqref{#1}}
\def\vech{\rm {vech}}

\def\vec{{\rm vec}}
\def\wt{ \widetilde }

\begin{document}

\title{Gaussian variational approximations for high-dimensional state space models}
\date{\empty}

\author[1]{Matias Quiroz}
\author[3,4]{David J. Nott}
\author[5]{Robert Kohn}
\affil[1]{University of Technology Sydney, School of Mathematical and Physical Sciences, Sydney NSW 2007, Australia.}
%\affil[2]{ARC Centre of Excellence for Mathematical and Statistical Frontiers (ACEMS), Melbourne VIC 3010, Australia.}
\affil[2]{Department of Statistics and Applied Probability, National University of Singapore, Singapore 117546.}
\affil[3]{Institute of Operations Research and Analytics, National University of Singapore, 21 Lower Kent Ridge Road,
Singapore 119077}\vspace{-1in}
\affil[4]{UNSW Business School, School of Economics, University of New South Wales, Sydney NSW 2052, Australia.}

\maketitle
\vspace{-0.3in}
\begin{abstract}
Our article considers a Gaussian variational approximation of the posterior density
 in a high-dimensional state space model. The variational parameters to be optimized
are the mean vector and the  covariance matrix of the approximation.
The number of parameters in the covariance matrix grows as the square of the number of model parameters,
so it is necessary to find simple yet effective parameterizations of the covariance structure when the
number of model parameters is large. We approximate the joint posterior distribution
over the high-dimensional state vectors by a dynamic factor model, having Markovian time dependence
and a factor covariance structure for the states. This gives a reduced description of the dependence structure for the
states, as well as a temporal conditional independence structure similar to that in the true posterior.
The usefulness of the approach is illustrated for prediction in two high-dimensional applications that are challenging for
Markov chain Monte Carlo sampling. The first is a spatio-temporal model for the spread of the Eurasian Collared-Dove across North America;
the second is a Wishart-based multivariate stochastic volatility model for financial returns.
%\smallskip

\end{abstract}

\section{Introduction}\label{sec:Intro}
Variational Approximation  (VA) \citep{Ormerod2010,blei+kj17} estimates the posterior distribution
of a model by  assuming the form for the posterior density and optimizing a measure of closeness
to the true posterior; e.g., a frequent choice for the approximation is a
multivariate Gaussian distribution, where the variational optimization is over an unknown mean and covariance matrix.
VA is becoming an increasingly popular way to estimate the posterior because of its
ability to handle large datasets and highly parameterized models. The
accuracy of the VA depends on a number of factors, such as
the flexibility of the approximating family, the model considered, and the sample size.
There are now some theoretical results which show that
the variational posterior converges to the true parameter value under suitable regularity conditions, and rates of convergence
have been established for parametric models \citep{Wang2018VBConsistency} and, more generally,
for non-parametric and high-dimensional models \citep{zhang2018convergence}.
However, for a finite number of observations, when the variational approximation does not collapse to a point mass,
it is often observed that there is a practically meaningful discrepancy between the uncertainty quantification
provided by the approximation and that of the true posterior distribution.
This is especially the case when the variational family used is insufficiently flexible.
Nevertheless, even in these cases, predictive inference -- predictions and prediction intervals -- obtained from VA seem
empirically to be usefully close to those obtained from the exact posterior.
As such, variational approximation methods provide a useful and fast alternative to Markov chain Monte Carlo (MCMC),
especially when predictive inference is the focus of the analysis.

\medskip

Our article considers a  Gaussian variational approximation (GVA) for a state space model when the state vector is high-dimensional.
Such models are common in spatio-temporal applications~\citep{cressie+w11}, financial econometrics~\citep{philipov+g06}, and in other important applications.
It is challenging to obtain  the GVA  when dealing with a high-dimensional model, because the number of variational parameters
in the covariance matrix of the VA grows quadratically with the number of model parameters. This makes
it necessary to parameterize the variational covariance matrix parsimoniously, but still be
able to capture the structure of the posterior.
This goal is best achieved by taking into account the structure of the posterior itself.  We do so by parameterizing
the variational posterior covariance matrix using a dynamic factor model, which reduces the dimension of the state vector.
The Markovian time dependence for the low-dimensional factors provides sparsity in the precision matrix for the factors.

\medskip

We develop efficient computational methods for forming the approximations and illustrate the advantages of the
approach in two high-dimensional example datasets. For both models, Bayesian inference by
 MCMC simulation is challenging.  The first is a spatio-temporal model for the spread of the Eurasian collared dove across North America \citep{wikle+h06};
the second is a multivariate stochastic volatility model for a collection of portfolios of assets \citep{philipov+g06}.
We derive GVAs for both models and show that they give useful predictive inference.

\medskip

VA estimates the posterior by optimization. Our article uses stochastic
gradient ascent methods to performing the optimization \citep{ji+sw10,nott+tvk12,Paisley2012,Salimans2013}.
In particular, the so-called reparameterization trick is used to unbiasedly  estimate the gradients of the variational objective
\citep{Kingma2013,rezende+mw14}. Section \ref{sec:stochasticgradientvariational} briefly reviews  these methods.
Applying these methods for GVA, \citet{Tan2016}  match the sparsity of the variational
precision matrix to the conditional independence structure of the true posterior based on a sparse Cholesky factor
for the precision matrix. Their motivation
is that zeros in the precision matrix of a Gaussian distribution correspond to conditional
independence between variables, and sparse matrix operations allow computations in the variational optimization to be done
efficiently. They apply their approach to both random effects models and state space models,
but their method is impractical for a state space model having a  high-dimensional state vector.
This approach and related approximations using a sparse precision matrix are reviewed further in Section~3.

\medskip

Our approach is also related to recent high-dimensional Gaussian posterior approximations having a factor structure.
Factor models \citep{Bartholomew2011} are well known to be useful for modeling dependence in high-dimensional settings.
\citet{ong+ns17} consider a Gaussian variational approximation for factor covariance structures using stochastic
gradient methods for the variational optimization.  Computation in the variational optimization can be done efficiently in
high dimensions using the Woodbury formula \citep{woodbury1950inverting}.
\citet{Barber1998} and \citet{Seeger2000} used factor structures in GVA, but only
when  the variational objective can be computed analytically or with one-dimensional quadrature.
\citet{rezende+mw14} consider a factor model for the precision matrix with one factor in some applications to
some deep generative models arising in machine learning applications.
\citet{Miller2016} considered factor parameterizations of covariance
matrices for normal mixture components
in a flexible variational boosting approximate inference method, including a method for exploiting the reparameterization
trick for unbiased gradient estimation of the variational objective in that setting.

\medskip

Various other parameterizations of the covariance matrix in GVA were considered by
\citet{Opper2009}, \citet{Challis2013} and \citet{Salimans2013}.  \citet{Salimans2013}
also considered efficient stochastic gradient methods for fitting such approximations,
using both gradient and Hessian information and exploiting other structure in the
target posterior distribution, as well as extensions to more complex hierarchical formulations including mixtures of normals.

\medskip
Our article makes use of both the conditional independence structure and the factor structure in forming the GVA for high dimensional state space models.
Bayesian computations for state space models are well known to be challenging for complex nonlinear models. It is usually feasible to carry out MCMC on a
complex state space model by sampling the states one at a time conditional on the neighbouring states
(e.g.,  \citealp{carlin1992monte});  in general, sampling one state at a time requires careful convergence diagnosis
and can fail if the dependence between states is strong. \cite{carter1994gibbs}  document this phenomenon in linear
Gaussian state space models, and we also document this problem of poor mixing for the spatio-temporal case \citep{wikle+h06} discussed later.

\medskip

State-of-the-art general
approaches using particle MCMC methods \citep{Andrieu2010} can in principle be much more efficient than
an MCMC algorithm that generates the states one at a time. However, particle MCMC is usually much slower than
MCMC because of the need to generate multiple particles at each time point.
Particle methods also have a number of other drawbacks, which depend on the model that is estimated. Thus, if
there is strong dependence between the states and parameters, then it is necessary to use pseudo-marginal methods
\citep{Beaumont2003, Andrieu2009} which estimate the likelihood and it is necessary to ensure that the variability
of the log of the estimated likelihood is sufficiently small \citep{pitt2012some, doucet2013}.
This is particularly difficult to do if the state dimension is high.

\medskip

The rest of the article is organized as follows.
Section~\ref{sec:stochasticgradientvariational} provides a brief description of variational approximation.
Section~\ref{sec:CovarianceParameterizations} reviews some previous parameterizations of the covariance matrix,
in particular the methods of \citet{Tan2016} and \citet{ong+ns17}
for GVA using conditional independence and a factor structure, respectively.
Section~\ref{sec:Gaussian_Variational_high_dim_state} describes our methodology,
which combines a factor structure for the states and conditional independence in time for the
factors to obtain flexible and convenient approximations of the posterior distribution
in high dimensional state space models.
Section~\ref{sec:examples} describes an extended example for a spatio-temporal dataset in ecology concerned with the spread of
the Eurasian collared-dove across North America.
Section~\ref{sec:examples2} considers variational inference in a Wishart based multivariate stochastic volatility model.
Appendix \ref{app:grad_expressions_lemmas} contains the necessary gradient expressions to implement our method.
Technical derivations and other details are placed in a self-contained supplement after the main article. We refer to equations, sections, etc in the main paper as (1), Section 1, etc, and in the supplement as (S1), Section S1, etc.

\section{Stochastic gradient variational methods\label{sec:stochasticgradientvariational}}
\subsection{Variational objective function\label{SS: var obj fn}}
Variational approximation methods \citep{Attias1999,Jordan1999,Winn2005} reformulate the problem of approximating an intractable
posterior distribution as an optimization problem.
Let $\theta=(\theta_1,\dots,\theta_d)^\top$ be the vector of model parameters, $y=(y_1,\dots,y_n)^\top$ the observations
and consider Bayesian inference for $\theta$ with a prior density $p(\theta)$.
Denoting the likelihood by $p(y|\theta)$, the posterior density is $p(\theta|y)\propto p(\theta)p(y|\theta)$, and in variational approximation we consider a family of densities $\{ q_\lambda(\theta)\}$, indexed by the variational parameter $\lambda$, to approximate $p(\theta|y)$. Our article takes the approximating family to be Gaussian so that $\lambda$ consists
of the mean vector and the distinct elements of the covariance matrix in the approximating normal density.

\medskip

To express the approximation of $p(\theta|y)$ as an optimization problem, we take the Kullback-Leibler (KL) divergence,
\begin{align*}
 \mathrm{KL}(q_\lambda(\theta)||p(\theta|y)) & = \int \log \frac{q_\lambda(\theta)}{p(\theta|y)}q_\lambda(\theta)\;d\theta,
\end{align*}
as the discrepancy measure between
$q_\lambda(\theta)$ to $p(\theta|y)$. The KL divergence is non-negative and zero if and only if $q_\lambda(\theta)=p(\theta|y)$.
It is straightforward to show that $\log p(y)$, where $p(y)=\int p(\theta)p(y|\theta)\;d\theta$, can be expressed as
\begin{align}
  \log p(y) & = {\cal L}(\lambda)+\mathrm{KL}(q_\lambda(\theta)||p(\theta|y)),  \label{logpy}
\end{align}
where
\begin{align}
  {\cal L}(\lambda) & = \int \log \frac{p(\theta)p(y|\theta)}{q_\lambda(\theta)}q_\lambda(\theta)\;d\theta \label{lb}
\end{align}
is referred to as the variational lower bound or evidence lower bound (ELBO).
We have that $\log p(y)\geq {\cal L}(\lambda)$, with
equality if and only if $q_\lambda(\theta)=p(\theta|y)$ because the KL divergence is non-negative.
\Eqref{logpy} shows that minimizing the KL divergence is equivalent to maximizing the ELBO in \eqref{lb} because
$\log p(y)$ does not depend on $\lambda$;  this is a more convenient optimization target
as it does not involve the intractable $p(y)$.

\medskip

For introductory overviews of variational methods for statisticians
see \citet{Ormerod2010, blei+kj17}.

\subsection{Stochastic gradient optimization\label{SS:Stochastic gradient optimization} }
Maximizing ${\cal L}(\lambda)$ to obtain an optimal approximation of $p(\theta|y)$ is often difficult in models with a non-conjugate prior structure, since ${\cal L}(\lambda)$
is defined as an integral which is generally intractable.  However, stochastic gradient methods
\citep{Robbins1951} are useful for performing the optimization
and there is now a large literature surrounding the application of this idea \citep[among others]{ji+sw10,Paisley2012,nott+tvk12,Salimans2013,Kingma2013,rezende+mw14,Hoffman2013,Ranganath2014,Titsias2015,Kucukelbir2016}.
In a simple stochastic gradient ascent method for optimizing ${\cal L}(\lambda)$, an initial guess for the optimal value $\lambda^{(0)}$ is updated according to the iterative
scheme
\begin{align}
 \lambda^{(t+1)} & =\lambda^{(t)}+a_t \widehat{\nabla_\lambda {\cal L}(\lambda^{(t)})}  \label{robinsm},
\end{align}
where $a_t$, $t\geq 0$ is a sequence of learning rates; $\nabla_\lambda {\cal L}(\lambda)$ is the gradient vector of ${\cal L}(\lambda)$ with respect
to $\lambda$; and $\widehat{\nabla_\lambda {\cal L}(\lambda)}$ denotes an unbiased estimate of $\nabla_\lambda {\cal
L}(\lambda)$.
The learning rate sequence is typically chosen to satisfy $\sum_t a_t=\infty$ and $\sum_t a_t^2<\infty$, which ensures that the
iterates $\lambda^{(t)}$ converge
to a local optimum as $t\rightarrow\infty$ under suitable regularity conditions \citep{bottou10}.
Various adaptive choices for the learning
rates are also possible and we use the ADADELTA \citep{Zeiler2012}
approach in our applications in Sections~\ref{sec:examples} and \ref{sec:examples2}.

\medskip
\subsection{Variance reduction\label{SS: Variance reduction}}
Application of stochastic gradient methods to variational inference depends on being able to obtain the required unbiased
estimates of the gradient of the lower bound in \eqref{robinsm}.
Reducing the variance of these gradient estimates as much as possible is important
for both the stability of the algorithm and fast convergence. Our article uses gradient estimates based on the so-called reparameterization trick
\citep{Kingma2013,rezende+mw14}. The lower bound ${\cal L}(\lambda)$ is an expectation with respect to $q_\lambda$,
\begin{align}
  {\cal L}(\lambda) & = E_q(\log h(\theta)-\log q_\lambda(\theta)), \label{lbq}
\end{align}
where $E_q(\cdot)$ denotes expectation with respect to $q_\lambda$ and $h(\theta)=p(\theta)p(y|\theta)$.  If we differentiate with respect to $\lambda$ under the integral sign in (\ref{lbq}), the
resulting expression for the gradient can also be written as an expectation with respect to $q_\lambda$, which is easily estimated unbiasedly by Monte Carlo integration provided that
sampling from this distribution is feasible. However, this
approach, called the score function method \citep{Williams1992}, typically has a very large variance. The reparameterization trick is often much more efficient \citep{xu2018some} and we now describe it.
Suppose that we can write
$\theta\sim q_\lambda(\theta)$ as $\theta=u(\lambda,\omega)$, where $\omega$ is a random vector with a density
which does not depend on the variational parameters $\lambda$, e.g. for a multivariate normal density $q_\lambda(\theta)=\mathcal{N}(\mu,\Sigma)$, with $\Sigma=CC^\top$, where
$C$ is the (lower triangular) Cholesky factor of $\Sigma$ we can write $\theta=\mu+C\omega$, where $\omega\sim \mathcal{N}(0,I_d)$ and  $I_d$ is the $d\times d$ identity matrix.
Substituting $\theta=u(\lambda,\omega)$ into (\ref{lbq}), we obtain
\begin{align}
  {\cal L}(\lambda) & = E_\omega(\log h(u(\lambda,\omega))-\log q_\lambda(u(\lambda,\omega))), \label{lbqr}
\end{align}
where $E_\omega$ is the expectation with respect to $\omega$.
Differentiating under the integral sign, we obtain
\begin{align}
  \nabla_\lambda {\cal L}(\lambda) & = E_\omega(\nabla_\lambda \log h(u(\lambda,\omega))-\nabla_\lambda \log q_\lambda(u(\lambda,\omega))), \label{lbqr2}
\end{align}
which  is easily estimated unbiasedly if it is possible to sample from $\omega$.

\medskip

We now discuss variance reduction beyond the reparameterization trick. \citet{Roeder2017}, generalizing arguments in \citet{Salimans2013}, \citet{Han2016} and \citet{Tan2016}, show that \eqref{lbqr2} can  be rewritten
as
\begin{align}
  \nabla_\lambda {\cal L}(\lambda) & = E_\omega\left(\frac{d u(\lambda,\omega)}{d\lambda}\left\{\nabla_\theta \log h(u(\lambda,\omega))-\nabla_\theta \log q_\lambda(u(\lambda,\omega))\right\}\right), \label{lbq2}
\end{align}
where $d u(\lambda,\omega)/d\lambda$ is defined as the matrix with element $(i,j)$ the partial derivative of the $i$th element of $u$ with respect to the $j$th element of
$\lambda$.  Note that if the approximation is exact, i.e. $q_\lambda(\theta)\propto h(\theta)$, then a Monte Carlo approximation to the expectation on the right hand
side of (\ref{lbq2}) is exactly zero even if such an approximation is formed using only a single sample from $f(\cdot)$.  This is one reason to prefer (\ref{lbq2})
as the basis for obtaining unbiased estimates of the gradient of the lower bound if the approximating variational family is flexible enough to provide
an accurate approximation.  However, \citet{Roeder2017} show that the extra terms that arise when (\ref{lbqr2}) is used directly for estimating the gradient
of the lower bound can be thought of as acting as a control variate, i.e. it reduces the variance, with a scaling that can be estimated empirically, although the computational cost of this
estimation may not be worthwhile.  In our state space model applications, we consider using both (\ref{lbqr2}) and (\ref{lbq2}), because our approximations may be very rough when the
dynamic factor parameterization of the variational covariance structure contains only a small number of factors.
Here, it may not be so relevant
to consider what happens in the case where the approximation is exact as a guide for reducing the variability of gradient estimates.

\medskip
\section{Parameterizing the covariance matrix\label{sec:CovarianceParameterizations}}

\subsection{Cholesky factor parameterization of $\Sigma$\label{subsec:cholesky_of_covariance}}

\citet{Titsias2014} considered normal variational posterior approximation using a Cholesky factor parameterization and used stochastic gradient methods
for optimizing the KL divergence.  \citet{Challis2013} also considered Cholesky factor parameterizations in Gaussian variational approximation,
but without using stochastic gradient optimization methods.

\medskip

For gradient estimation, \citet{Titsias2014} consider the reparameterization trick with
$\theta=\mu+C\omega$, where $\omega \sim \mathcal{N}(0,I_d)$, $\mu$ is the variational posterior mean and
$\Sigma=CC^\top$ is the variational posterior covariance with lower triangular Cholesky factor $C$ and with the diagonal elements of $C$ being positive. Hence, $\lambda$ contains $\mu$ and the non-zero elements of $C$ and (\ref{lbqr}) becomes, apart from terms not depending on $\lambda$,
\begin{align}\label{eq:ELBO_chol_param_Sigma}
{\cal L}(\lambda) & = E_\omega(\log h(\mu+C\omega))+ \log |C|,
\end{align}
and note that $\log |C|=\sum_i \log C_{ii}$ since $C$ is lower triangular. \citet{Titsias2014} derive the gradient of \eqref{eq:ELBO_chol_param_Sigma}, and  it is straightforward to estimate the expectation $E_\omega$ unbiasedly by simulating one or more samples $\omega$ and computing their average, i.e. plain Monte Carlo integration. The method can also be considered in conjunction with data subsampling. \citet{Kucukelbir2016} considered a similar approach.
\medskip

\subsection{Sparse Cholesky factor parameterization of $\Omega=\Sigma^{-1}$}

\citet{Tan2016}
  consider an approach which parameterizes the precision matrix $\Omega=\Sigma^{-1} = CC^\top$
  in terms of its Cholesky  factor $C$, and impose a sparse structure on $C$ which comes
  from the conditional independence structure in the model.
To minimize notation, we continue to write $C$ for a Cholesky factor used to parameterize the variational posterior even
though here it is the Cholesky factor of the precision matrix rather than of the covariance matrix as in the previous subsection.
Similarly to \citet{Tan2016}, \citet{Archer2016} also consider parameterizing a Gaussian variational approximation using
the precision matrix, but they optimize directly with respect to the elements $\Omega$, while also exploiting sparse matrix computations in obtaining the Cholesky factor of $\Omega$. \cite{Archer2016} are also concerned with
state space models and impose a block tridiagonal structure on the variational posterior precision matrix for the states, using functions of local
data parameterized by deep neural networks to describe blocks of the mean vector and precision matrix corresponding to different states.
Recently \cite{spantini+bm19} have considered variational algorithms for filtering and smoothing based on transport maps;  they
also consider online approaches to estimation of the fixed parameters.
\medskip

Here, we follow \citet{Tan2016} and parameterize  the variational approximation in terms of
the Cholesky factor $C$ of $\Omega$.
Section~\ref{sec:Gaussian_Variational_high_dim_state} shows how to use
the conditional independence structure in the model  to impose a sparse structure on $C$.
We note that sparsity is very important for reducing the number of variational parameters that need to be optimized,
so that a sparse $C$ allows the Gaussian variational approximation method to be extended to high-dimensions.

Using  the reparameterization trick, with $q_\lambda(\theta)=\mathcal{N}(\mu,C^{-\top}C^{-1})$, implies that
$\theta=\mu+C^{-\top}\omega$, with $\omega \sim \mathcal{N}(0,I_d)$. Here,
$C^{-\top}:=(C^{-1})^\top$ and $\lambda:=(\mu^\top,\vech(C)^\top)^\top$ where $\vech(C)$ is the half-vectorization of $C$ stacking
the elements of $C$ below the diagonal in a vector going from left to right.

\medskip

Similarly to
Section~\ref{subsec:cholesky_of_covariance},
\begin{align*}
 {\cal L}(\lambda) & = E_\omega(\log h(\mu+C^{-\top}\omega)-\log q_\lambda(\mu+C^{-\top}\omega)),
\end{align*}
which, apart from terms not depending on $\lambda$, is
\begin{align}\label{eq:ELBO_tan_and_nott}
 {\cal L}(\lambda) & = E_\omega(\log h(\mu+C^{-\top}\omega))-\log |C|;
\end{align}
with $\log |C|=\sum_i \log C_{ii}$ since $C$ is lower triangular. \citet{Tan2016} derive
 the gradient of
\eqref{eq:ELBO_tan_and_nott} and, moreover, consider some improved gradient estimates for which \citet{Roeder2017}
provide a more general understanding. Section~\ref{sec:Gaussian_Variational_high_dim_state} applies
\citeauthor{Roeder2017}'s approach to our methodology.

\medskip

\subsection{Latent factor parameterization of $\Sigma$\label{subsec:OngEtAl_review}}

While the method of \citet{Tan2016} is an attractive way to reduce the number of variational parameters in problems with an
exploitable conditional independence structure, there are models where no such structure is available.
An alternative
parsimonious parameterization  is to use a factor structure
\citep{Geweke1996,Bartholomew2011}.  \citet{ong+ns17} parameterize the variational posterior
covariance matrix
$\Sigma:=BB^\top + D^2$, where $B$ is a $d\times q$ matrix with $q\ll d$, $B_{ij}=0$ for $i<j$, and $D$ is a diagonal matrix
with diagonal elements $\delta=(\delta_1,\dots, \delta_d)^\top$. The variational posterior becomes $q_\lambda(\theta)=
\mathcal{N}(\mu,BB^\top+D^2)$ with $\lambda=(\mu,B,\delta)$; this corresponds to the generative
model $\theta=B\omega+\delta\odot\kappa$ with $(\omega,\kappa)\sim \mathcal{N}(0,I_{d+q})$, where $\odot$ denotes elementwise
multiplication.
\citet{ong+ns17} applied the reparameterization trick based on this transformation and derive gradient
expressions of the resulting evidence lower bound.
\citet{ong+ns17} also outline how to efficiently implement the computations.
 Section \ref{subsec:EfficientComp} discusses this further.

 \medskip

 \section{Methodology\label{sec:Gaussian_Variational_high_dim_state}}
\subsection{Model, prior and posterior\label{SS: model and prior}}
Let $y=(y_1,\dots,y_T)^\top$ be an observed time series, generated by the state space model
\begin{subequations}\label{eq: SS model}
\begin{align}
  y_t|X_t & =x_t  \sim  m_t(y|x_t,\zeta), \;\;\;\;t=1,\dots, T,  \\
  X_t|X_{t-1} &=x_{t-1}  \sim s_t(x|x_{t-1},\zeta),   t=1,\dots, T;
\end{align}
\end{subequations}
where the prior density for $X_0$ is $p(X_0|\zeta)$,
$\zeta$ are the unknown fixed (non-time-varying) parameters in the model,
and the elements of $\zeta$ in the measurement and the state equation are typically different,
but the same symbol is used for brevity.
The observations $y_t$ are conditionally independent given the states $X=(X_0^\top,\dots, X_T^\top)^\top$, and
the prior distribution of $X$ given $\zeta$ is
$$p(X|\zeta)=p(X_0|\zeta)\prod_{t=1}^T s_t(X_t|X_{t-1},\zeta).$$
Let $\theta=(X^\top, \zeta^\top)^\top$ denote the full set of unknowns in the model.

\medskip
The posterior density of $\theta$ is
$p(\theta|y)\propto p(\theta)p(y|\theta)$,
with $p(\theta)=p(\zeta)p(X|\zeta)$, where $p(\zeta)$ is the prior density for $\zeta$ and $p(y|\theta)=\prod_{t=1}^T m_t(y_t|X_t,\zeta)$. Let $p$ be the dimension of $X_t$ and suppose $p$ is large. Approximating the joint posterior distribution in this setting is difficult and Section~\ref{subsec:structure_variational_posterior}  describes  a method based on GVA.

\medskip

\subsection{Example: Multivariate Stochastic Volatility\label{SS: wishart example}}

We illustrate some of the above ideas with the multivariate stochastic volatility model introduced by \citet{philipov+g06},
who used it to model the time-varying  dependence of a portfolio of $k$ assets over $T$ time periods;
Section~\ref{sec:examples2} discusses the model in more detail.

\citeauthor{philipov+g06} assume that the return at time period $t$, $t=1, \dots, T$, is the vector
$y_t=(y_{t1},\dots, y_{tk})^\top$,
\begin{subequations}\label{eq: mv sv example}
\begin{align}
y_t & \sim  \mathcal{N}(0, \Sigma_t), \quad \Sigma_t \in\mathbb{R}^{p\times p} \\
\Sigma_t^{-1} &\sim  \mathrm{Wishart}(\nu,S_{t-1}), \quad S_t=\frac{1}{\nu}  H(\Sigma_t^{-1})^d H^\top , \, S_t \in\mathbb{R}^{p\times p},\, \nu > k, \, 0 < d < 1;
\end{align}
\end{subequations}
$H$ is an unknown positive definite matrix and $\nu, d$ and $k$ are unknown scalars; $\Sigma_0$ is a known positive definite matrix. Section~\ref{sec:examples2} describes the priors for all the
fixed parameters and latents.

The state vector in this model is $\vech(\Sigma_t)$, and has dimension $p(p+1)/2$.
It is very high dimensional when  $p$ is large, e.g. it is 55 dimensional for $p=10$. It may then  be
necessary to use particle methods,
which are typically very slow, to estimate such a high dimensional model.
Section~\ref{sec:examples2} gives a more complete discussion.

\medskip

\subsection{Structure of the variational approximation\label{subsec:structure_variational_posterior}}
The variational posterior density $q_\lambda(\theta)$ for $\theta$, is based on a generative model which has the dynamic factor structure,
\begin{align}
 X_t & =Bz_t+\epsilon_t\;\;\;\;\epsilon_t\sim \mathcal{N}(0,D_t^2),   \label{ldsm}
\end{align}
where $B$ is a $p\times q$ matrix, $q\ll p$, and $D_t$ is a diagonal matrix with diagonal elements $\delta_t=(\delta_{t1},\dots,\delta_{tp})^\top$.
Let $z=(z_0^\top,\dots, z_T^\top)^\top$ and
$\rho=(z^\top,\zeta^\top)^\top\sim \mathcal{N}(\mu,\Sigma)$, $\Sigma=C^{-\top}C^{-1}$ where $C$ is the Cholesky factor of the precision matrix of
$\rho$.  We will write $q$ for the dimension of each $z_t$, with $q \ll p = \dim(X_t)$, and assume that
$$C=\left[\begin{array}{cc}
 C_1 & 0 \\\
 0 & C_2
\end{array}\right],$$
is block diagonal;
$C_1$ is the Cholesky factor of the precision matrix $\Omega_1 = C_1 C_1^\top$ for $z$;
 and $C_2$ is the Cholesky factor for the precision matrix of $\zeta$.
Let  $\Sigma_1$ denote the covariance matrix of $z$. We further assume that
$C_1$ is lower triangular with a single band,  implying that $\Omega_1$ is band tridiagonal.
See Section~\ref{app:SparsityPrecisionMatrix} of the supplement for details.
For a Gaussian distribution, zero elements in the precision matrix represent conditional independence relationships.
In particular, the sparse structure imposed on $C_1$ means that in the generative distribution for $\rho$,
the latent variable $z_t$, given $z_{t-1}$ and $z_{t+1}$, is conditionally independent of the remaining
elements of $z$;  in other words, if we think of the variables $z_t$, $t=1,\dots,T$ as a time series,
they have a Markovian dependence structure.

We now construct the variational distribution for $\theta$ through
\begin{align*}
  \theta & = \left[\begin{array}{cc} X \\ \zeta \end{array}\right] = \left[\begin{array}{cc} I_{T+1}\otimes B & 0 \\ 0 & I_P\end{array}\right]\rho + \left[\begin{array}{c} \epsilon \\ 0 \end{array}\right],
\end{align*}
where $\otimes$ denotes the Kronecker product, $P$ is the dimension of $\zeta$, and $\epsilon=(\epsilon_0^\top,\dots,\epsilon_T^\top)^\top$.
We can apply the reparameterization trick by writing $\rho=\mu+C^{-\top}\omega$, where $\omega\sim \mathcal{N}(0,I_{q(T+1)+P})$. Then,
\begin{align}
\theta & = W\rho+Ze= W\mu+WC^{-\top}\omega+Ze,  \label{reparformula}
\end{align}
where
$$W=\left[\begin{array}{cc} I_{T+1}\otimes B & 0_{p(T+1)\times P} \\ 0_{P\times q(T+1)} & I_P \end{array}\right];  \;\;\;\;Z=\left[\begin{array}{cc} D & 0_{p(T+1)\times P} \\ 0_{P\times p(T+1)} & 0_{P\times P} \end{array}\right],\;\;\;\; e=\left[\begin{array}{c} \epsilon \\ 0_{P\times 1} \end{array}\right], $$
$D$ is a diagonal matrix with diagonal entries $(\delta_0^\top,\dots,\delta_T^\top)^\top$, and
$u=(\omega^\top,\epsilon^\top)^\top \sim \mathcal{N}(0,I_{(p+q)(T+1)+P})$.
We also write $\omega=(\omega_1^\top,\omega_2^\top)^\top$, where the blocks of this partition follow those of
$\rho=(z^\top,\zeta^\top)^\top$.

\medskip

The factor model  above describes  the covariance structure for the states, as well as
for dimension reduction in the variational posterior mean of the states, since $E(X_t)=B\mu_t$, where $\mu_t=E(z_t)$.  An alternative is to set $E(z_t)=0$ and use
\begin{align}
 X_t & = \mu_t+Bz_t+\epsilon_t, \label{hdsm}
\end{align}
where $\mu_t$ is now a $p$-dimensional vector specifying the variational posterior mean for $X_t$ directly.

\medskip

We call parameterization \eqref{ldsm} the low-dimensional state
mean (LD-SM) parameterization, and parameterization \eqref{hdsm} the high-dimensional state mean (HD-SM) parameterization. In both parameterizations, $B$ forms a basis for $X_t$, which is reweighted over time according to the latent weights (factors) $z_t$. The LD-SM parameterization provides information on how these basis functions are reweighted over time to form the  approximate posterior mean, since $E(X_t)=B\mu_t$ and we infer both $B$ and $\mu_t$ in the variational optimization.
Section~\ref{sec:examples} illustrates this basis representation. Appendix \ref{app:grad_expressions_lemmas} outlines the gradients and their derivation for the LD-SM parameterization. Derivations for the HD-SM parameterization follow those for the LD-SM case with straightforward minor adjustments.

\medskip

Algorithm~\ref{Alg:variational_optimization_dynamic_factor_approximation} outlines the stochastic gradient ascent algorithm that
maximizes \eqref{lbqr}. Lemmas~\ref{lem:standard_gradient} and \ref{lem:Roeder_gradient} in
Appendix~\ref{app:grad_expressions_lemmas} obtain
the gradients.
Their expectations are estimated by one or more samples from $u$. We note that
although automatic differentiation implementations have improved enormously in recent years, and there
are some implementations of linear algebra operators supporting sparse precision matrices
\citep{durrande19}, it is not straightforward to use automatic differentiation for the structured
matrix manipulations necessary for efficient computation here which make use of a combination of sparse and low rank
matrix computations.

The gradients are computed  by
either \eqref{gradmuss}-- \eqref{gradCss} in Lemma~\ref{lem:standard_gradient}, or by
equations \eqref{sgradmuss},
\eqref{sgradBss}, \eqref{sgraddss} and \eqref{sgradCss} in Lemma~\ref{lem:Roeder_gradient}.
If the variational approximation to the posterior is accurate, then
\eqref{sgradmuss}, \eqref{sgradBss}, \eqref{sgraddss} and \eqref{sgradCss} corresponding to the gradient estimates recommended in \citet{Roeder2017} may be prefered;  Section~\ref{sec:stochasticgradientvariational} explains the reasons.
  However, since we consider massive dimension reduction with only a small numbers of factors the approximation may be crude and we therefore investigate both approaches in later examples.

\medskip
Finally,  it is well known that factor models have identifiability issues \citep{shapiro1985identifiability}.  The choice of identifying constraints
in factor models can matter, particularly for interpretation.
However, here the choice of any identifying constraints is not crucial as
we do not interpret either the factors or the loadings, but only use them for modeling the covariance matrix and,
in the LD-SM parameterization, also the variational mean.  Factor structures are widely used
as a method for achieving parsimony in the model formulation in the state space framework for spatio-temporal data \citep{wikle+c99, lopes2008},
multivariate stochastic volatility \citep{ku2014flexible, philipov2006factor}, and in other applications \citep{aguilar+w00,carvalho2008}.
This is distinct from the main idea in the present paper of using a dynamic factor structure
for dimension reduction in a variational approximation for getting parsimonious but flexible descriptions of
dependence in the posterior for approximate inference.

\medskip

\begin{algorithm}[tbh]
	\caption{Stochastic gradient ascent for optimizing the variational objective $\mathcal{L}(\lambda)$ in \eqref{lbqr}. See Appendix \ref{app:grad_expressions_lemmas} for notation and gradients.} %\hspace{\textwidth} \~ The The sampled momentum (in brackets) is not needed, but opens up some algorithmic possibilities, for example correlating the momentum between iterations \citep{neal2011mcmc}. }	
	\SetKwInOut{Input}{Input}
	\vspace{1mm}
	\Input {Starting values $\lambda_0 \leftarrow (\mu_0, B_0, \delta_0, C_0)$, learning rates $\eta_\mu, \eta_B, \eta_\delta, \eta_C$, number of iterations $M$.}
	\vspace{1mm}
		\vspace{1mm}
	\For{$m = 1$ \KwTo $M$} {
%						\vspace{1mm}
%		$\theta_{m-1} = W_{m-1}\mu^{} \lambda_m  \leftarrow \lambda_{m-1} + \eta \widehat{\nabla\mathcal{L}}(\lambda^{m-1})$ \Comment{a comment} \\
%							\vspace{1mm}
		$\mu_m  \leftarrow \mu_{m-1} + \eta_\mu \odot \widehat{\nabla_\mu\mathcal{L}}(\lambda_{m-1})$ \Comment{$\nabla_{\mu}\mathcal{L}$ in \eqref{gradmuss} or \eqref{sgradmuss}} \\
							\vspace{1mm}
 $\lambda_{m-1} \leftarrow  (\mu_{m}, B_{m-1}, \delta_{m-1}, C_{m-1})$ \Comment{Update $\mu$}\\
						\vspace{1mm}
		$B_m  \leftarrow B_{m-1} + \eta_B \odot \widehat{\nabla_{\mathrm{vec}(B)}\mathcal{L}}(\lambda_{m-1})$ \Comment{$\nabla_{\mathrm{vec}(B)}\mathcal{L}$ in \eqref{gradBss} or \eqref{sgradBss}} \\
							\vspace{1mm}
 $\lambda_{m-1} \leftarrow  (\mu_{m}, B_{m}, \delta_{m-1}, C_{m-1})$ \Comment{Update $B$}\\
						\vspace{1mm}
		$\delta_m  \leftarrow \delta_{m-1} + \eta_\delta \odot \widehat{\nabla_{\delta}\mathcal{L}}(\lambda_{m-1})$ \Comment{$\nabla_{\delta}\mathcal{L}$ in \eqref{graddss} or \eqref{sgraddss}} \\
							\vspace{1mm}
 $\lambda_{m-1} \leftarrow  (\mu_{m}, B_{m}, \delta_{m}, C_{m-1})$ \Comment{Update $\delta$}\\
						\vspace{1mm}
		$C_m  \leftarrow C_{m-1} + \eta_C \odot \widehat{\nabla_C\mathcal{L}}(\lambda_{m-1})$ \Comment{$\nabla_{C}\mathcal{L}$ in \eqref{gradCss} or \eqref{sgradCss}} \\
							\vspace{1mm}
 $\lambda_{m} \leftarrow  (\mu_{m}, B_{m}, \delta_{m}, C_{m})$ \Comment{Update $C$}\\

							\vspace{1mm}
 $\lambda_{m-1} \leftarrow \lambda_{m}$  \Comment{Update $\lambda$}\\
	}
	\vspace{1mm}

	\textbf{Output:} $\lambda_m$
	
\label{Alg:variational_optimization_dynamic_factor_approximation}
\end{algorithm}

\medskip

\subsection{Efficient computation}\label{subsec:EfficientComp}
The gradient estimates for the lower bound (see Appendix \ref{app:grad_expressions_lemmas} for expressions) are efficiently computed using a combination of sparse
matrix operations (for evaluating terms such as $C^{-\top}\omega$ and the high-dimensional matrix multiplications in the expressions) and, as in \citet{ong+ns17}, the Woodbury identity for dense matrices such as $(W\Sigma W^\top+Z^2)^{-1}$ and $(W_1\Sigma_1W^\top+D^2)^{-1}$. The Woodbury identity is
\begin{align*}
  (\Lambda \Gamma \Lambda^\top + \Psi)^{-1} = & \Psi^{-1}-\Psi^{-1}\Lambda(\Lambda^\top\Psi^{-1}\Lambda + \Gamma^{-1})^{-1}\Lambda^\top \Psi^{-1}
\end{align*}
for conformable matrices $\Lambda, \Gamma$ and diagonal $\Psi$; it reduces the required computations into a much lower dimensional space since $q \ll p$ and  $\Psi$ is diagonal.

\medskip

\section{Application 1: Spatio-temporal model}\label{sec:examples}
\subsection{Eurasian collared-dove data}

The first example considers the spatio-temporal model of \citet{wikle+h06}
for a dataset on the spread of the Eurasian collared-dove across North America. The dataset consists of the
number of doves $y_{s_{i}t}$ observed at location $s_{i}$ (latitude,
longitude) $i=1,\dots,p,$ in year $t=1,\dots,T=18,$ corresponding
to an observation period of 1986-2003. The spatial locations
correspond to
$p=111$ grid points with the dove counts aggregated
within each area. See \citet{wikle+h06} for details.
The count observed at location $s_{i}$ at time $t$ depends on the
number of times $N_{s_{i}t}$ that the location was sampled. However, this
variable is unavailable and therefore we set the offset in the model
to zero, i.e. $\log(N_{s_{i}t})=0$.

\medskip

\subsection{Model\label{subsec:Model}}

Let $y_t=(y_{s_1t},\dots, y_{s_pt})^\top$ denote the count data at time $t$.  \citet{wikle+h06} model $y_t$ as conditionally independent Poisson
variables, where the log means are given by a latent high-dimensional Markovian process $u_t$ plus measurement error.  The dynamic process $u_t$ evolves according to a discretized
diffusion equation;  specifically,
the model in \citet{wikle+h06} is
\begin{align*}
y_{t}|v_{t} & \sim\mathrm{Poisson}(\mathrm{diag}(N_{t})\exp(v_{t}))\quad y_{t},N_{t},v_{t}\in\mathbb{R}^{p}\\
v_{t}|u_{t},\sigma_{\epsilon}^{2} & \sim \mathcal{N}(u_{t},\sigma_{\epsilon}^{2}I_{p}),\quad u_{t}\in\mathbb{R}^{p},I_{p}\in\mathbb{R}^{p\times p},\sigma_{\epsilon}^{2}\in\mathbb{R}^{+}\\
u_{t}|u_{t-1},\psi,\sigma_{\eta}^{2} & \sim \mathcal{N}(H(\psi)u_{t-1},\sigma_{\eta}^{2}I_{p}),\quad\psi\in\mathbb{R}^{p},H(\psi)\in\mathbb{R}^{p\times p},\sigma_{\eta}^{2}\in\mathbb{R}^{+},
\end{align*}
with priors $\sigma_{\epsilon}^{2},\sigma_{\psi}^{2},\sigma_{\alpha}^{2}\sim\mathrm{IG}(2.8,0.28),\sigma_{\eta}^{2}\sim\mathrm{IG}(2.9,0.175)$
and
\begin{align*}
u_{0} & \sim \mathcal{N}(0,10I_{p})\\
\psi|\alpha,\sigma_{\psi}^{2} & \sim \mathcal{N}(\Phi\alpha,\sigma_{\psi}^{2}I_{p}),\quad\Phi\in\mathbb{R}^{p\times l}\text{,}\alpha\in\mathbb{R}^{l},\sigma_{\psi}^{2}\in\mathbb{R}^{+}\\
\alpha & \sim \mathcal{N}(0,\sigma_{\alpha}^{2}R_{\alpha}),\quad\alpha_{0}\in\mathbb{R}^{l},R_{\alpha}\in\mathbb{R}^{l\times l},\sigma_{\alpha}^{2}\in\mathbb{R}^{+}.
\end{align*}
$\mathrm{Poisson(\cdot)}$ is the Poisson distribution for a (conditionally)
independent response vector parameterized in terms of its expectation
and $\mathrm{IG}(\cdot)$ is the inverse-gamma distribution with shape
and scale as arguments. The spatial dependence is modeled via the
prior mean $\Phi\alpha$ of the diffusion coefficients $\psi$, where
$\Phi$ consists of the $l$ orthonormal eigenvectors with the largest
eigenvalues of the spatial correlation matrix $R(c)=\exp(-cd)\in\mathbb{R}^{p\times p}$,
where $d$ is the Euclidean distance between pairwise grid locations
in $s_{i}$. Finally, $R_{\alpha}$ is a diagonal matrix with the
$l$ largest eigenvalues of $R(c)$. We follow \citet{wikle+h06}
and set $l=1$ and $c=4$.

\medskip

Let $u=(u_{0}^{\top},\dots u_{T}^{\top})^{\top}$ , $v=(v_{1}^{\top},\dots v_{T}^{\top})^{\top}$
and denote the parameter vector
\[
\theta=(u,v,\psi,\alpha,\log\sigma_{\epsilon}^{2},\log\sigma_{\eta}^{2},\log\sigma_{\psi}^{2},\log\sigma_{\alpha}^{2}),
\]
which we infer through the posterior
\begin{eqnarray}
p(\theta|y) & \propto & \sigma_{\epsilon}^{2}\sigma_{\eta}^{2}\sigma_{\psi}^{2}\sigma_{\alpha}^{2}p(\sigma_{\epsilon}^{2})p(\sigma_{\eta}^{2})p(\sigma_{\psi}^{2})p(\sigma_{\alpha}^{2})p(\alpha|\sigma_{\alpha}^{2})p(\psi|\alpha,\sigma_{\psi}^{2})\nonumber \\
 & ~ & p(u_{0})\prod_{t=1}^{T}p(u_{t}|u_{t-1},\psi,\sigma_{\eta}^{2})p(v_{t}|u_{t},\sigma_{\epsilon}^{2})p(y_{t}|v_{t}),\label{eq:PosteriorDistribution}
\end{eqnarray}
with $y=(y_{1}^{\top},\dots,y_{T}^{\top})^{\top}$. Section \ref{app:GradientLogPosterior} of the supplement derives the gradient of the log-posterior required by the variational Bayes (VB) approach.

\medskip

\subsection{Variational approximations of the posterior distribution\label{subsec:VA_posterior_dist}}
Section~\ref{sec:Gaussian_Variational_high_dim_state} considers two different parameterization of the low rank
approximation, in which either the state vector $X_{t}$ has mean
$E(z_{t})=B\mu_{t}$, $\mu_{t}\in\mathbb{R}^{q}$ (low-dimensional
state mean, LD-SM) or $X_{t}$ has a separate mean $\mu_{t}\in\mathbb{R}^{p}$
and $E(z_{t})=0$ (high-dimensional state mean, HD-SM). In this particular
application there is a third choice of parameterization which
we now consider.

\medskip

The model in Section~\ref{subsec:Model} connects the data with the high-dimensional state vector $u_{t}$ via a high-dimensional
auxiliary variable $v_{t}$. In the notation of Section~\ref{sec:Gaussian_Variational_high_dim_state},
we can include $v$ in $\zeta$, in which
case the parameterization of the variational posterior is the one described there.
We refer to this parameterization as a low-rank state (LR-S).
However, it is clear from \eqref{eq:PosteriorDistribution} that there
is posterior dependence between $u_{t}$ and $v_{t}$, but the variational
approximation in Section~\ref{sec:Gaussian_Variational_high_dim_state} omits the dependence between $z$ and $\zeta$.
Moreover, $v_{t}$ is also high-dimensional, but the LR-S parameterization
does not reduce its dimension. An alternative parameterization that
deals with both considerations includes $v$ in the $z$-block, which
we refer to as the low-rank state and auxiliary variable (LR-SA) parameterization. This comes at
the expense of omitting dependence between $v_{t}$ and $\sigma_{\epsilon}^{2}$,
but also becomes more computationally costly because, while the total
number of variational parameters is smaller (see Table \ref{tab:VB_parameterization} in Section \ref{app:Sparsity} of the supplement),
the dimension of the $z$-block increases ($B$ and $C_{1}$) and
the main computational effort lies here and not in the $\zeta$-block.
Table
\ref{tab:ELBO_final_iteration} shows the CPU times relative to the LR-S parameterization. The LR-SA parameterization requires
a small modification of the derivations in Section \ref{sec:Gaussian_Variational_high_dim_state}, which we outline
in detail in Section~\ref{sec:LRSA_derivations} of the supplement
as they can be useful for other models with a high-dimensional auxiliary variable.

\medskip

It is straightforward to deduce conditional independence relationships
in \eqref{eq:PosteriorDistribution} to build the Cholesky factor
$C_{2}$ of the precision matrix $\Omega_{2}$ of $\zeta$ in Section
4, with
\[
\zeta=\begin{cases}
(v,\psi,\alpha,\log\sigma_{\epsilon}^{2},\log\sigma_{\eta}^{2},\log\sigma_{\psi}^{2},\log\sigma_{\alpha}^{2}) & \text{(LR-S)}\\
(\psi,\alpha,\log\sigma_{\epsilon}^{2},\log\sigma_{\eta}^{2},\log\sigma_{\psi}^{2},\log\sigma_{\alpha}^{2}) & \text{(LR-SA)}.
\end{cases}
\]
Section \ref{sec:Gaussian_Variational_high_dim_state} outlines the construction of the Cholesky factor $C_{1}$ of the
precision matrix $\Omega_{1}$ of $z$, whereas
the minor modification needed for LR-SA is in Section \ref{sec:LRSA_derivations} of the supplement. We note that, regardless of the parameterization, we obtain massive parsimony (between $6,428\text{-}11,597$ variational parameters) compared to the saturated Gaussian variational approximation which in this application has $8,923,199$ parameters; see Section \ref{app:Sparsity}
of the supplement for further details.

We consider four different variational parameterizations, combining
each of LR-SA or LR-S with the different parameterization of the means
of $X_{t}$, i.e. LD-SM or HD-SM. In all cases, we let $q=4$ and
perform $10,000$ iterations of a stochastic optimization algorithm with
learning rates chosen adaptively according to the ADADELTA approach
\citep{Zeiler2012}. We use the gradient estimators in \citeauthor{Roeder2017}, i.e. \eqref{sgradmuss}, \eqref{sgradBss}, \eqref{sgraddss} and \eqref{sgradCss}, although we found no noticeable difference compared to (\ref{gradmuss}) -- (\ref{gradCss}); it is likely that this is due   to the small number of factors as described in Sections \ref{sec:stochasticgradientvariational} and \ref{sec:Gaussian_Variational_high_dim_state}. Our choice was motivated by computational efficiency as some terms cancel out using the approach in \citeauthor{Roeder2017}. We initialize
 $B$ and $C$ as
unit diagonals and, for parity, $\mu$ and $D$ are chosen to match the starting values of the Gibbs sampler
in \citeauthor{wikle+h06}.

\medskip

Figure~\ref{fig:ELBOs} monitors the convergence via the estimated
value of ${\cal L}(\lambda)$ using a single Monte Carlo sample. Table \ref{tab:ELBO_final_iteration}
presents estimates of ${\cal L}(\lambda)$ at the final iteration using $100$ Monte Carlo samples.
The results suggest that the best VB parameterization in terms of ELBO is the low-rank state algorithm (LR-SA) with, importantly, a high-dimensional state-mean (HD-SM) (otherwise the poorest VB approximation is achieved, see Table \ref{tab:ELBO_final_iteration}). However, Table \ref{tab:VB_parameterization} shows that this parameterization is about three times as CPU intensive. The fastest VB parameterizations are both Low-Rank State (LR-S) algorithms, and modeling the state mean separately for these
does not seem to improve ${\cal L}(\lambda)$ (Table \ref{tab:ELBO_final_iteration})
and is also slightly more computationally expensive (Table \ref{tab:VB_parameterization}). Taking these considerations into account, the final choice of VB parameterization we use
for this model
is the low-rank state with low-dimensional state mean (LR-S + LD-SM). Section~\ref{subsec:Results} shows that this parameterization gives accurate approximations for our analysis. For the rest of this example, we benchmark the VB posterior from LR-S + LD-SM against the MCMC approach in \citeauthor{wikle+h06}.

\begin{figure}[h]
\centering
\includegraphics[width=0.3\columnwidth]{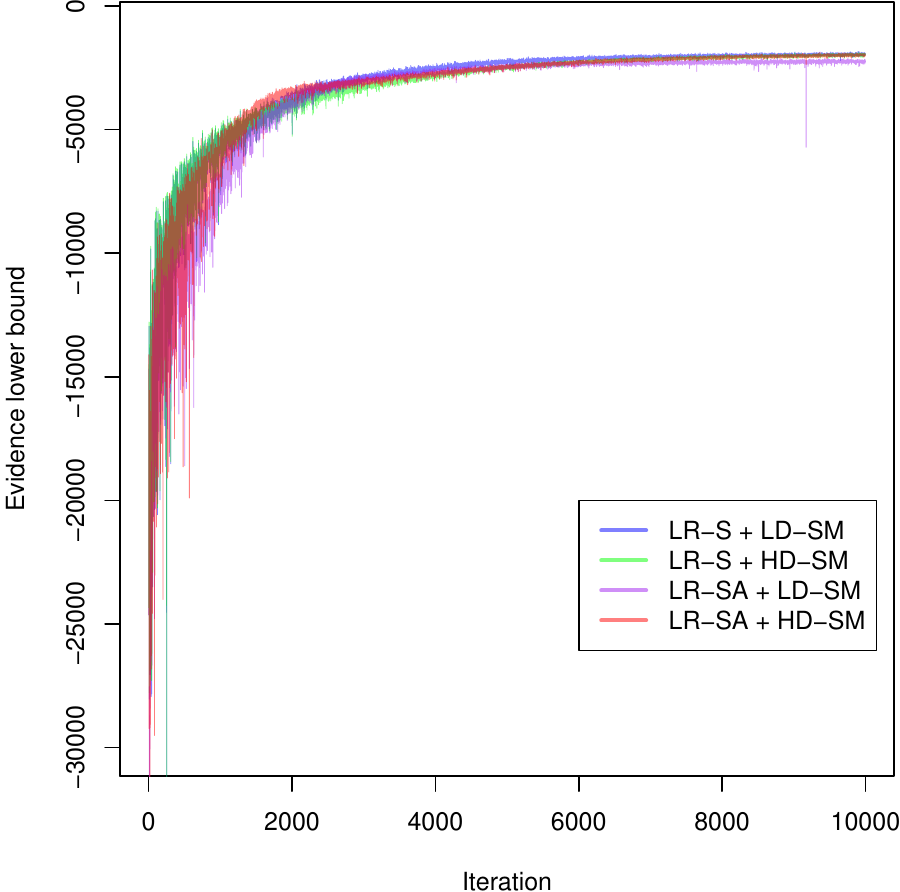}

\caption{${\cal L}(\lambda)$ for the variational approximations for the spatio-temporal example. The figure shows
the estimated value of ${\cal L}(\lambda)$ vs iteration number for the four different parameterizations, see Section \ref{subsec:VA_posterior_dist} or Table \ref{tab:ELBO_final_iteration} for abbreviations.}
\label{fig:ELBOs}
\end{figure}

\begin{table}[h]
\centering \caption{${\cal L}(\lambda)$ and CPU time for the VB parameterizations in the spatio-temporal and Wishart process example. The table shows the estimated value of ${\cal L}(\lambda)$ for the different VB parameterizations by combining low-rank state / low-rank state and auxiliary (LR-S / LR-SA)
with either of low-dimensional state mean / high-dimensional state
mean (LD-SM / HD-SM). The estimate and its $95$\% confidence interval
are computed at the final iteration using $100$ Monte Carlo samples. The table also show the relative CPU (R-CPU) times, where the reference is LD-SM.
}

\begin{tabular}{llclclcc}
\toprule
{\footnotesize{}{}}\textbf{\footnotesize{}parameterization}  &  &  &  &  &  &  & \tabularnewline
\cmidrule{1-1}
\textbf{\footnotesize{}{}}\emph{\footnotesize{}Spatio-temporal} &  & {\footnotesize{}{}R-CPU}  &  & {\footnotesize{}{}$\mathcal{L}(\lambda_{\mathrm{opt}})$}  &  & {\footnotesize{}{}Confidence interval}  & \tabularnewline
\cmidrule{1-1} \cmidrule{3-3} \cmidrule{5-5} \cmidrule{7-7}
{\footnotesize{}{}LR-S + LD-SM}  &  & {\footnotesize{}{}$1$}  &  & {\footnotesize{}{}$\textnormal{-}1,996$}  &  & {\footnotesize{}{}$[\textnormal{-}2,004;\textnormal{-}1,988]$}  & \tabularnewline
{\footnotesize{}{}LR-S + HD-SM}  &  & {\footnotesize{}{}$1.005$}  &  & {\footnotesize{}{}$\textnormal{-}2,024$}  &  & {\footnotesize{}{}$[\textnormal{-}2,032;\textnormal{-}2,016]$}  & \tabularnewline
{\footnotesize{}{}LR-SA + LD-SM}  &  & {\footnotesize{}{}$3.189$}  &  & {\footnotesize{}{}$\textnormal{-}2,158$}  &  & {\footnotesize{}{}$[\textnormal{-}2,167;\textnormal{-}2,148]$}  & \tabularnewline
{\footnotesize{}{}LR-SA + HD-SM}  &  & {\footnotesize{}{}$3.017$}  &  & {\footnotesize{}{}$\textnormal{-}1,909$}  &  & {\footnotesize{}{}$[\textnormal{-}1,918;\textnormal{-}1,900]$}  & \tabularnewline
 &  &  &  &  &  &  & \tabularnewline
\cmidrule{1-1}
\textbf{\footnotesize{}{}}\emph{\footnotesize{}Wishart process} &  &  &  &  &  &  & \tabularnewline
\cmidrule{1-1}
{\footnotesize{}{}LR-S + LD-SM}  &  & {\footnotesize{}{}$1$}  &  & {\footnotesize{}{}$\textnormal{-}1,121$}  &  & {\footnotesize{}{}$[\textnormal{-}1,126;\textnormal{-}1,115]$}  & \tabularnewline
{\footnotesize{}{}LR-S + HD-SM}  &  & {\footnotesize{}{}$1.0004$}  &  & {\footnotesize{}{}$\textnormal{-}1,040$}  &  & {\footnotesize{}{}$[\textnormal{-}1,046;\textnormal{-}1,035]$}  & \tabularnewline
\bottomrule
\end{tabular}\label{tab:ELBO_final_iteration}

\end{table}

\medskip

\subsection{MCMC settings\label{SS: mcmc settings}}

Before comparing VB to MCMC, it is necessary to determine a reasonable
burn-in period and number of iterations for inference for the Gibbs sampler in \citeauthor{wikle+h06}.
It is clear that it is infeasible to monitor convergence for every
single parameter in such a large scale model as \eqref{eq:PosteriorDistribution},
and therefore we focus on $\psi$, $u_{18}$ and $v_{19}$, which
are among the variables considered in the analysis in Section~\ref{subsec:Results}.

\medskip

\citeauthor{wikle+h06} use $50,000$ iterations of which $20,000$
are discarded as burn-in.
We generate $4$ sampling chains with these settings and inspect convergence
using the $\mathtt{coda}$ package \citep{plummer2006coda} in $\mathtt{R}$.
We compute the Scale Reduction Factors (SRF) \citep{gelman1992inference}
for $\psi,u_{18}$ and $v_{19}$ as a function of the number of Gibbs
iterations. The adequate number of iterations in MCMC depends on what functionals of the parameters are of interest;  here
 we monitor convergence for these quantities
since we report marginal posterior distributions for these quantities later.
The scale reduction factor of a parameter measures if
there is a significant difference between the variance within the
four chains and the variance between the four chains of that parameter.
We use the rule of thumb that concludes convergence when $\mathrm{SRF}<1.1$,
which gives a burn-in period of approximately $40,000$ here, for these functionals. After discarding
these samples and applying a thinning of $10$ we are left with $1,000$
posterior samples for inference. However, as the draws are auto-correlated,
this does not correspond to $1,000$ independent draws used in the
analysis in Section \ref{subsec:Results} (note that we obtain independent
samples from our variational posterior). To decide how many Gibbs
samples are equivalent to $1,000$ independent samples for $\psi,u_{18}$
and $v_{19}$, we compute the Effective Sample Size (ESS) which takes
into account the auto-correlation of the samples. We find that the
smallest is $\mathrm{ESS}=5$ and hence we require $200,000$ iterations
after a thinning of $10$, which makes a total of $2,000,000$ Gibbs
iterations, excluding the burn-in of $40,000$.  Thinning is advisable here due
to memory issues --- it is impractical to store $2,000,000$ iterations for each parameter (which may be used, for example, to estimate kernel densities) in high-dimensional models.

\subsection{Analysis and results\label{subsec:Results}}
We first consider inference on the diffusion coefficient $\psi_i$ for location $i$. Figure \ref{fig:Distribution_delta} shows the ``true'' posterior (represented by MCMC) together with the variational approximation for six locations described in the caption of the figure. The figure shows that the posterior distribution is highly skewed for locations with zero dove counts and approaches normality as the dove counts increase. Consequently, the accuracy of the variational posterior (which is Gaussian) improves with increasing dove counts. The figure also shows the phenomena we described in the beginning of Section \ref{sec:Intro}: there is a discrepancy between the posterior densities. However, as we will see, it affects neither the location estimates of the intensity of the process nor its prediction.

\begin{figure}[H]
\centering
\includegraphics[width=0.5\columnwidth]{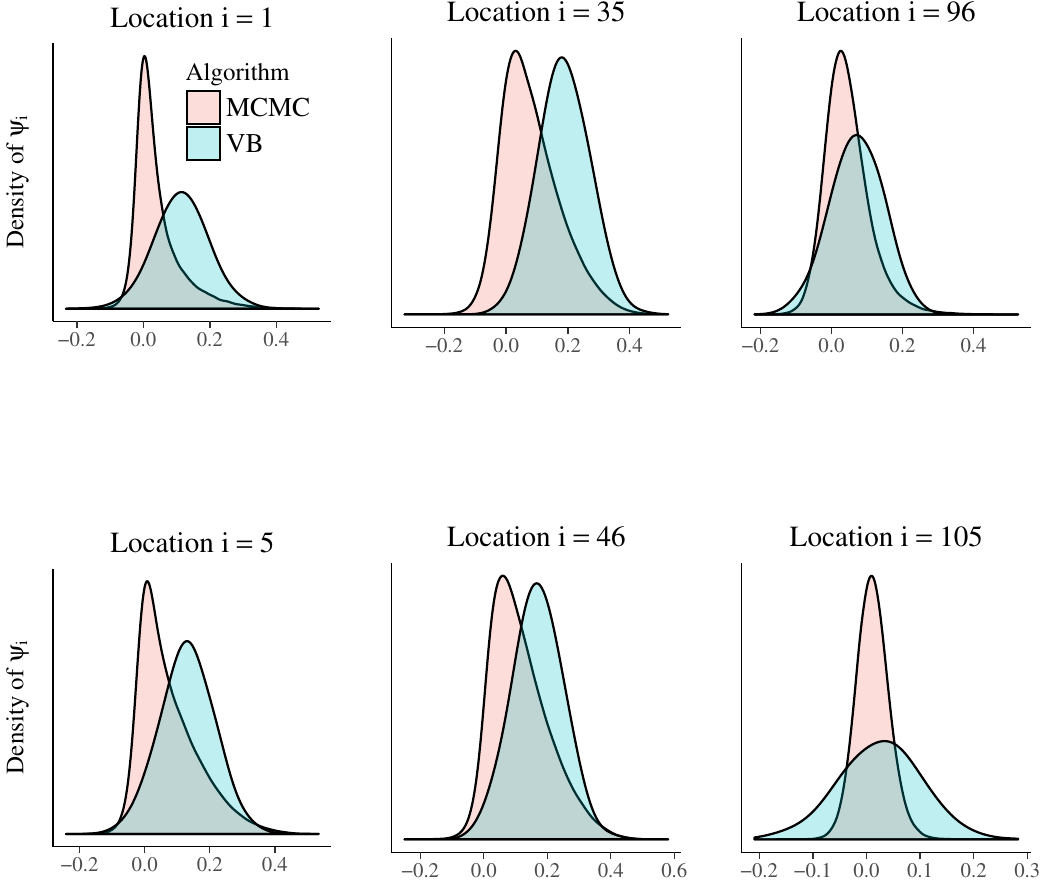}
\caption{Distribution of the diffusion coefficients.  The figure shows the posterior distribution of $\psi_i$ obtained by MCMC and VB. The locations are divided into three categories (total doves over time within brackets): zero count locations (Idaho, $i=1\, [0]$ , Arizona $i=5\,[0]$, left panels), low count locations (Texas, $i=35\,[16], 46\,[21]$, middle panels) and high count locations (Florida, $i=96\,[1,566], 105\,[1,453]$, right panels).}
\label{fig:Distribution_delta}
\end{figure}

\begin{figure}[H]
\centering
\includegraphics[width=0.5\columnwidth]{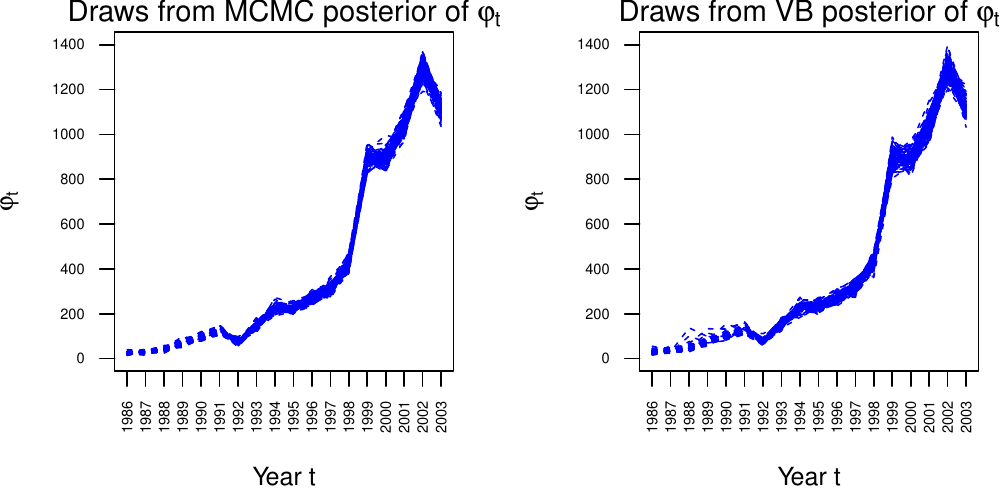}

\caption{Samples from the posterior sum of dove intensity over the spatial
grid for each year. The figure shows $100$ samples from the posterior
distribution of $\varphi_{t}=\sum_{i}\exp(v_{it})$ obtained by MCMC
(left panel) and VB (right panel).}
\label{fig:lambda_each_year}
\end{figure}

\medskip

Figure\emph{ }\ref{fig:lambda_each_year} shows 100 VB and MCMC posterior
samples of the dove intensity for each year summed over the spatial
locations, i.e. $\varphi_{t}=\sum_{i}\exp(v_{it})$. The two posteriors
are similar and show an exponential increase
of doves until the year $2002$ followed by a steep decline for $2003$.

Figure~\ref{fig:dove_intensity_MCMC_vs_VB} summarises some spatial properties of the model.
It shows a heat map of the MCMC
and VB posterior means of the dove intensity $\varphi_{it}=\exp(v_{it})$
for the last five years of the data, overlaid on a map of the U.S.
The figure confirms that VB gives accurate location estimates of the spatial process in this example.

\begin{figure}[h]
\centering
\includegraphics[width=0.6\columnwidth]{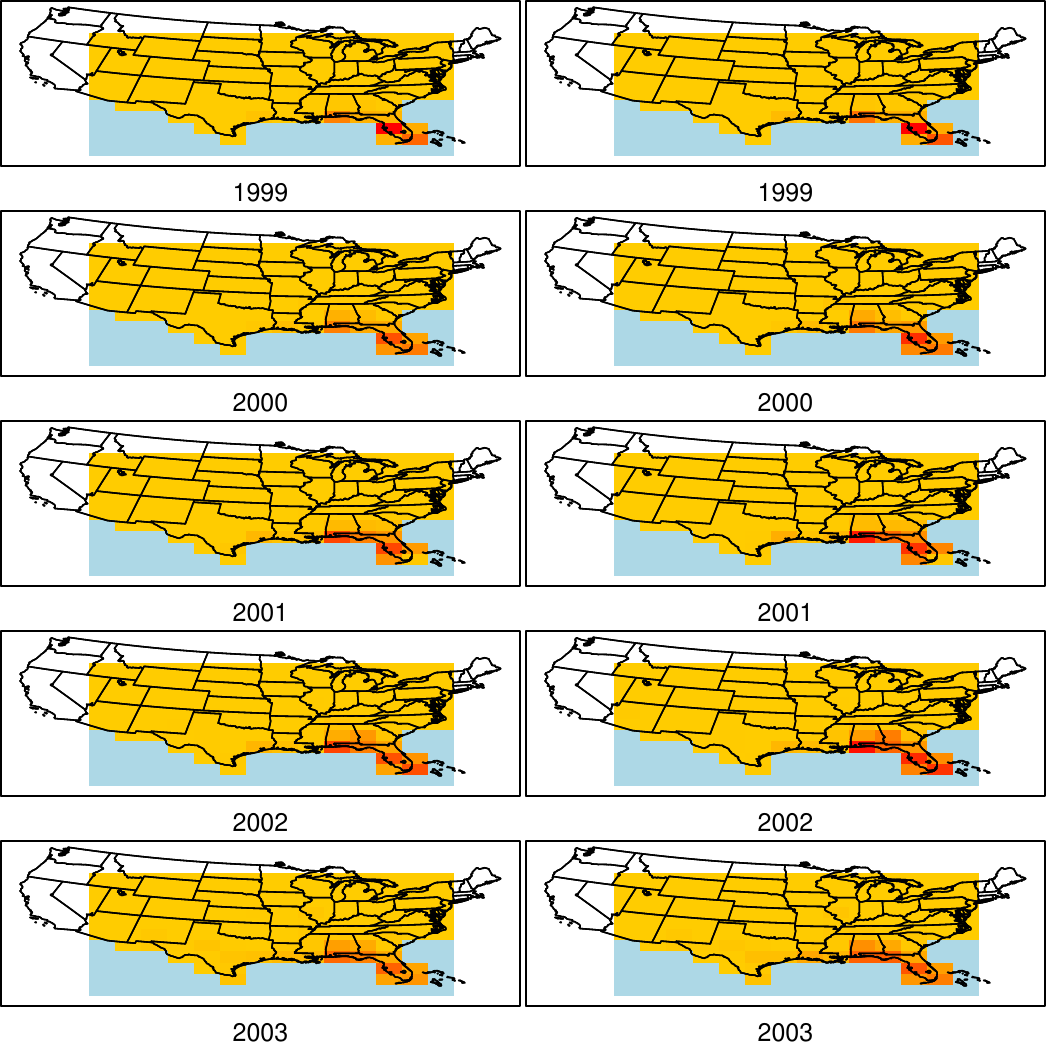}

\caption{Posterior dove intensity for the years 1999-2003. The figure
shows the posterior mean of $\varphi_{it}=\exp(v_{it})$ computed
by MCMC (left panels) and VB (right panels) for $i=1,\dots,p=111,$
and the last $5$ years of the data ($t=14,15,16,17,18$). The results
are illustrated with a spatial grid plotted together with a map of
the United States, where the colors vary between low intensity (yellow)
and high intensity (red). The light blue color is for aesthetic reasons
and does not correspond to observed locations. }
\label{fig:dove_intensity_MCMC_vs_VB}
\end{figure}

\medskip

We draw the following conclusions from the analysis
using the MCMC and VB posteriors, which are nearly identical. First,
Figure~\ref{fig:dove_intensity_MCMC_vs_VB} shows that the
the dove intensity is most pronounced in the South East states, in
particular Florida.
Second, Figure~\ref{fig:dove_intensity_MCMC_vs_VB} also shows that
it is likely that the decline of doves for year $2003$ in Figure \ref{fig:lambda_each_year}
can be  attributed to a drop in the intensity
at two areas of Florida: Central Florida ($i=96$) and South East
Florida ($i=105$). Figure \ref{fig:log_intensity_in_and_out_of_sample}
illustrates the whole posterior distribution of the log-intensity
for these locations at the year $2003$, as well as an out-of-sample
posterior predictive distribution for year $2004$.  Both estimates
are kernel density estimates using approximately $1,000$
effective samples. The posterior distributions for the VB and MCMC
are similar, and it is evident that using this large scale model for
forecasting future values is associated with a large uncertainty.

\medskip

Figure~\ref{fig:Spatial_basis_and_weightning} illustrates
 the spatial basis functions and their reweighting over time to produce mean of the variational approximation, as discussed in Section \ref{subsec:structure_variational_posterior}.

\begin{figure}[h]
\centering
\includegraphics[width=0.45\columnwidth]{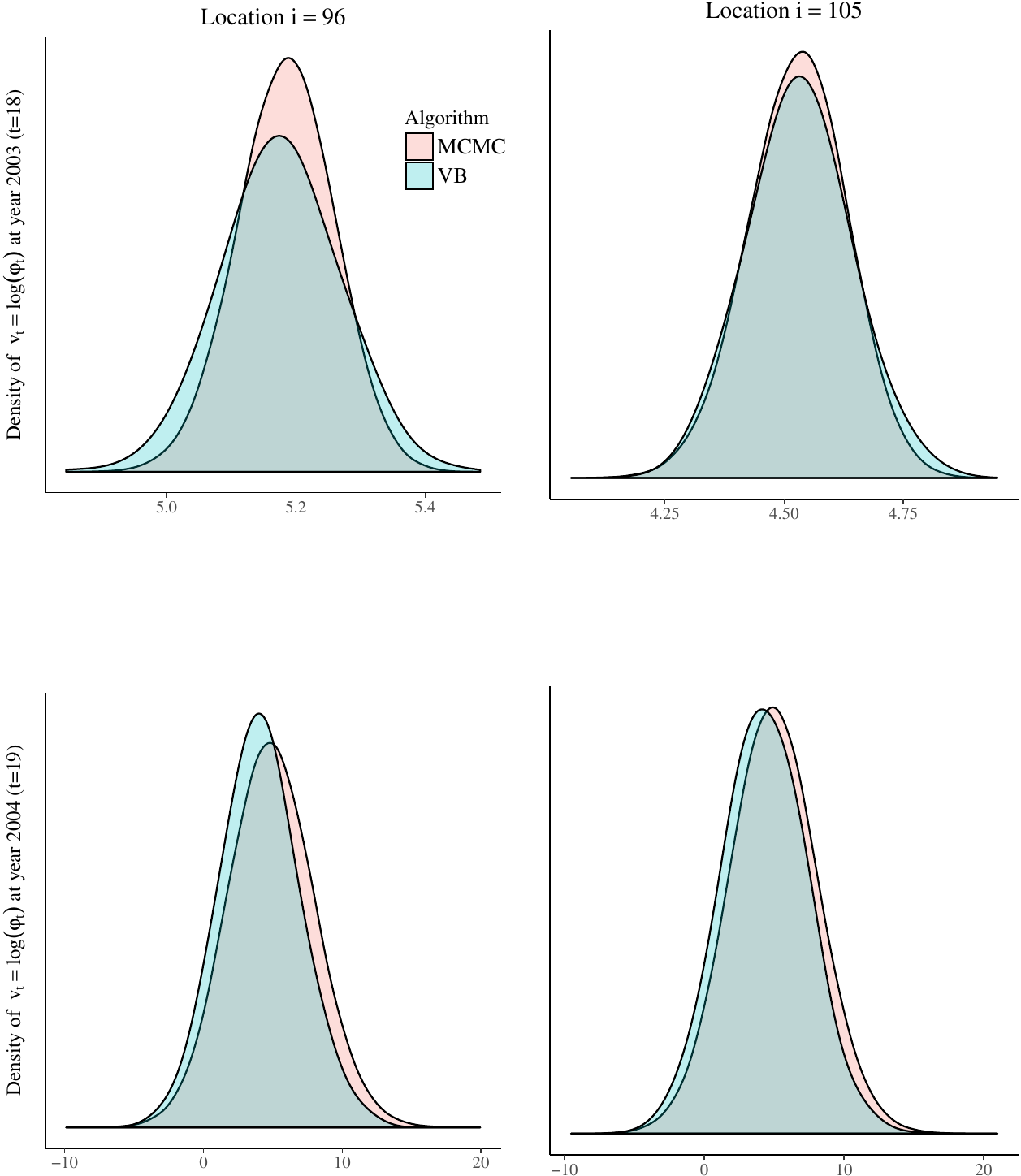}
\caption{Forecasting the log-intensity of the spatial process. The figure
shows an in-sample forecast density of the log-intensity $v_{it}$ for year
2003 ($t=18$, upper panels) and out-of-sample forecast density for year 2004
($t=19$, lower panels) for Central Florida ($i=96$, left panels)
and South East Florida ($i=105$, right panels).}
\label{fig:log_intensity_in_and_out_of_sample}
\end{figure}

\begin{figure}[t]
\centering
\includegraphics[width=0.5\columnwidth]{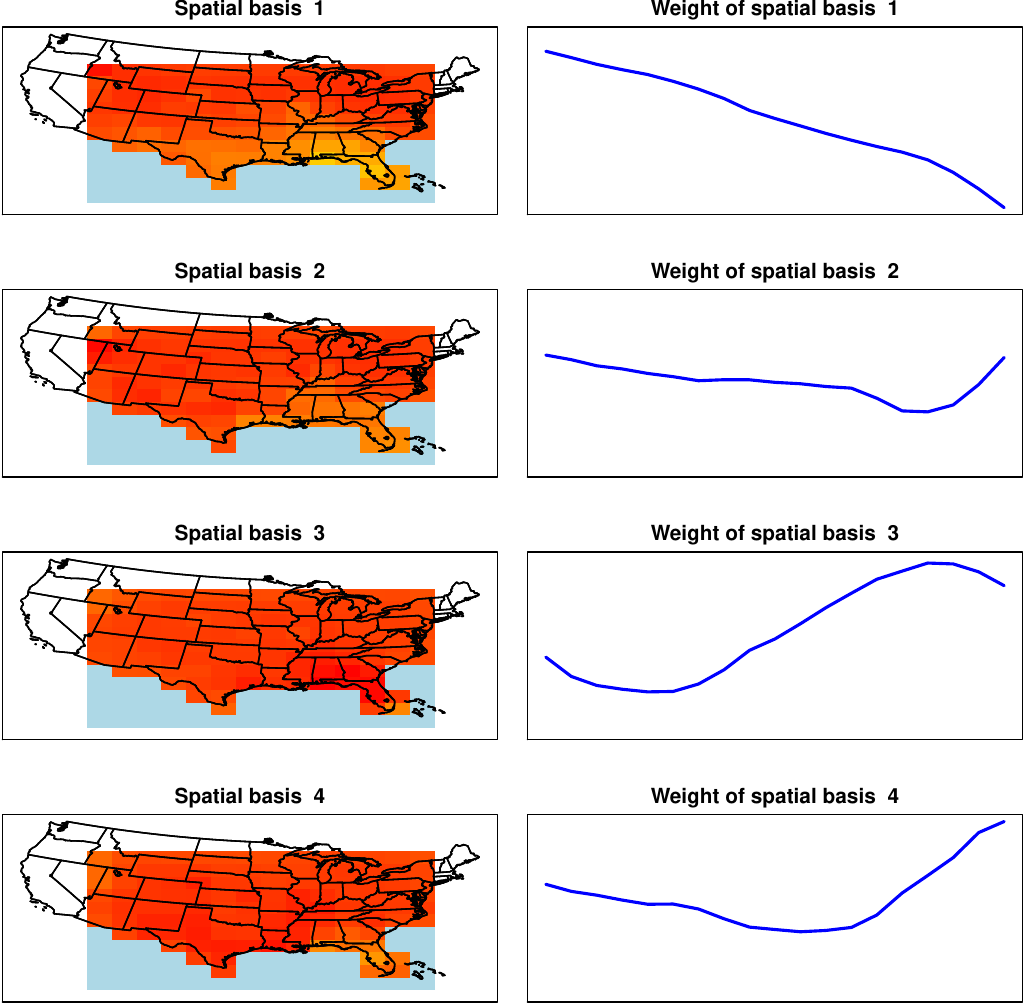}

\caption{Spatial basis representation of the state vector. The figure
shows the Spatial basis functions (left panel), i.e. the $j$th column of $B$, $j=1,\dots,q=4$ and the corresponding weights $\mu_t$ (right panel) through $t = 0, \dots, 18$, that fors $E(X_t) = B\mu_t$.} %[MQ: I am struggling with R to add y and x ticks to the weights, i.e. the right panel]}
\label{fig:Spatial_basis_and_weightning}
\end{figure}

\section{Application 2: Stochastic volatility modeling}\label{sec:examples2}

%show that the variational approximation provides a useful alternative. For estimation, we use $T=100$ monthly observations on value-weighted portfolios from the 201709 CRSP database, covering a period from 2009-06 to 2017-09. We follow \citet{philipov+g06} and prefilter each series using an AR(1) process.

%We remark that other MCMC approaches for estimating this model more efficiently might be possible, but it is outside the scope of this paper to pursue this. As an example, Hamiltonian Monte Carlo on the Riemannian manifold \citep{girolami2011riemann} has proven to be effective in sampling models with $500\text{-}1,000$ parameters. However, the computational burden relative to standard MCMC is increased and, moreover, tuning the algorithm becomes more difficult.

\subsection{Model\label{subsec:ModelPhilipovGlickman}}
The  second example considers the  Wishart based multivariate stochastic volatility model proposed
in \citet{philipov+g06} who used it to model the time-varying dependence of a portfolio of
$k$ assets over $T$ time periods. Section~\ref{SS: wishart example} briefly discussed this model.
\medskip

\citet{philipov+g06}
 assume that the return at time period $t$, $t=1, \dots, T$, is the vector
$y_t=(y_{t1},\dots, y_{tk})^\top$, with
\begin{eqnarray*}
\label{eq:MV_SV_model}
y_t & \sim & \mathcal{N}(0, \Sigma_t); \quad \Sigma_t \in\mathbb{R}^{p\times p} ;\\
\Sigma_t^{-1} &\sim & \mathrm{Wishart}(\nu,S_{t-1}); \quad S_t=\frac{1}{\nu}  H(\Sigma_t^{-1})^d H^\top ; \, S_t \in\mathbb{R}^{p\times p};\, \nu > k; \, 0 < d < 1;
\end{eqnarray*}
$H$ is a lower triangular Cholesky factor of a positive definite matrix $A$, with  $A=HH^\top \in \mathbb{R}^{p\times p}$; and $\Sigma_0$ is assumed known. The prior for $A$ is inverse Wishart,
i.e. $A^{-1}\sim \mathrm{Wishart}(\gamma_0,Q_0)$, with $\gamma_0=k+1$ and  $Q_0=I$;
a uniform prior on $[0,1]$ for
$d$, i.e.  $d\sim U[0,1]$; and a shifted gamma prior for $\nu$, i.e. $\nu-k\sim \mathrm{Gamma}(\alpha_0,\beta_0)$. The joint posterior density for $(\Sigma,A,\nu-k,d)$ is
\begin{align}
 p(\Sigma,A,\nu-k,d|y) & \propto p(A,d,\nu-k) \prod_{t=1}^T p(\Sigma_t|\nu,S_{t-1})p(y_t|\Sigma_t);\label{jointpost}
\end{align}
$p(A,d,\nu-k)$ is the joint prior density for $(A,d,\nu-k)$; $p(\Sigma_t|\nu,S_{t-1},d)$ is the conditional inverse Wishart prior
for $\Sigma_t$ given $\nu$, $S_{t-1}$; and $d$, and $p(y_t|\Sigma_t)$ is the normal density for $y_t$ given $\Sigma_t$.

\medskip

We write
$C_t$ for the Cholesky factor of $\Sigma_t$ and
we reparameterize the posterior in terms of the unconstrained parameter vector
$$\theta =(\mathrm{vech}(H')^\top,d',\nu',\mathrm{vech}(C_1')^\top,\dots, \mathrm{vech}(C_T')^\top)^\top; $$
where
\begin{eqnarray*}
C_t'  \in  \mathbb{R}^{k \times k}; & ~ & C'_{t,ij}=C_{t,ij}; \, i\neq j, \text{ and } C'_{t,ii}=\log C_{t,ii};\\
H'  \in  \mathbb{R}^{k \times k}; & ~ & H'_{ij}=H_{ij}; \, i\neq j, \text{ and } H_{ii}=\log H_{ii};
\end{eqnarray*}
with $d'=\log d/(1-d)$ and $\nu'=\log (\nu-k)$. Section~\ref{app:LogPosteriorPhilipovGlickman} shows that
\begin{align}
\label{eq:PosteriorDistribution_PhilipovGlickman}
  p(\theta|y) \propto & |L_k (I_{k^2}+K_{k,k})(H\otimes I_k)L_k^\top| \times \left\{\prod_{t=1}^T |L_k(I_{k^2}+K_{k,k})(C_t\otimes I_k)L_k^\top| \right\} \times (\nu-k) \\ \nonumber
 & \times d(1-d)\times \left\{\prod_i H_{ii}\right\}\left\{\prod_{t=1}^T \prod_{i=1}^k C_{t,ii}\right\} \times p(A,d,\nu-k)\left\{\prod_{t=1}^T p(\Sigma_t|\nu,S_{t-1},d) p(y_t|\Sigma_t)\right\};
\end{align}
Section \ref{supp:notation} of the supplement defines the elimination matrix $L_k$ and the commutation matrix $K_{k,k}$;
Section \ref{app:GradientLogPosteriorPhilipovGlickman} of the supplement shows how to evaluate the gradient of the log posterior.

\subsection{Evaluating the predictive
performance of the variational approximation\label{subsec:data_PhilipovGlickman}}
\citet{philipov+g06}
 develop an MCMC algorithm to estimate their Wishart based multivariate stochastic volatility model.
\citet{rinnergschwenter+tw12} point out that the Gibbs sampler developed by \citeauthor{philipov+g06} contains a mistake
which affects all the full conditionals.  We find that
implementing the \lq corrected\rq{}  version of their algorithm
results in a very inefficient sampler even for  the $k = 5$ portfolios used
by \citeauthor{philipov+g06} in their empirical example.
This means that the \lq corrected\rq{} \citeauthor{philipov+g06} algorithm
cannot be used as a \lq gold standard\rq{} to compare against the variational approximation
 results.  Although it
 may be possible to estimate the posterior of \citeauthor{philipov+g06}'s model using particle methods,
 we do not pursue this here. Section~\ref{subsec:ProblemsMCMC_PhilipovGlickman} of the supplement
illustrates the inefficiency of the corrected \citet{philipov+g06} sampler and explains its problems.

\medskip

We now show empirically (by simulation)
 that the variational posterior provides useful predictive inference. Since MCMC is unavailable,
the  GVA  is benchmarked against an oracle predictive approach, which assumes the
the static model parameters are known. We use a bootstrap particle filter \citep{gordon1993novel}
to obtain the posterior density of the state-vector at $t = T$;  it is then possible to
obtain the one-step ahead oracle predictive density $p(y_{T+1}|y_{1:T}, \zeta^{\mathrm{true}})$,
where $\zeta^{\mathrm{true}}$ denotes the true static model parameters.
The variational predictive density is then benchmarked against the oracle predictive density; we note that
the variational predictive density  integrates over the variational posterior of
all the parameters, including the static model parameters.

\medskip

Section~\ref{subsec:OraclePredictive} of the supplement
shows how to simulate from the oracle predictive density. Section~\ref{subsec:VariationalPredictive} of the supplement
 shows how to simulate from the variational predictive density. The  one-step ahead prediction is repeated
 for $H=4$ horizons. At horizons $h=1,\dots, H$, both filtering densities are based on $y_{1:T+h-1}$ and
 the optimization for finding the variational posterior for $h>1$ is fast since the
 variational parameters are initialized (except the ones added at $T+h$) at their variational mode from the previous optimization.

\medskip

\subsection{Variational approximations of the posterior distribution}
Since this example does not include a high-dimensional auxiliary variable, we use the low-rank state (LR-S) parameterization combined with both a low-dimensional state mean (LD-SM) and a high-dimensional state mean (HD-SM).
As in the previous example, it is straightforward to deduce conditional independence relationships in \eqref{eq:PosteriorDistribution_PhilipovGlickman} to build the Cholesky factor
$C_{2}$ of the precision matrix $\Omega_{2}$ of $\zeta$ in Section~\ref{sec:Gaussian_Variational_high_dim_state};
this section also outlines how to construct
 the Cholesky factor $C_{1}$ of the
precision matrix $\Omega_{1}$ of $z$. Massive parsimony is achieved in this application.
 In particular, for $k=12$ assets,  the saturated Gaussian variational approximation has $31,059,020$ parameters, while our parameterization has $10,813$. For $k=5$, the saturated case has $1,152,920$ parameters and our parameterizations has $4,009\text{-}5,109$. See Section~\ref{app:Sparsity} of the supplement for more details.

\medskip

For all variational approximations we let $q=4$ and
perform $10,000$ iterations of a stochastic optimization algorithm with
learning rates chosen adaptively according to the ADADELTA approach
\citep{Zeiler2012}. We initialize $B$ and $C$ as
unit diagonals and choose $\mu$ and $D$ randomly.
Figure~\ref{fig:ELBOs_PhilipovGlickman} monitors the estimated ELBO for both parameterizations,
using both the gradient estimators in \citeauthor{Roeder2017} and the alternative standard ones
which do not cancel terms that have zero expectation. For $k=5$, the figure shows that the different gradient
estimators perform equally well. Moreover, slightly more variable estimates are observed
in the beginning for the low-dimensional state mean parameterization compared to that of the
high-dimensional mean. Table~\ref{tab:ELBO_final_iteration} presents estimates of ${\cal L}(\lambda)$
at the final iteration using $100$ Monte Carlo samples and also presents the relative CPU times of the
algorithms. In this example, the separate state mean present in the high-dimensional state mean seems
to improve the ELBO considerably.

\begin{figure}[h]
\centering

\includegraphics[width=0.6\columnwidth]{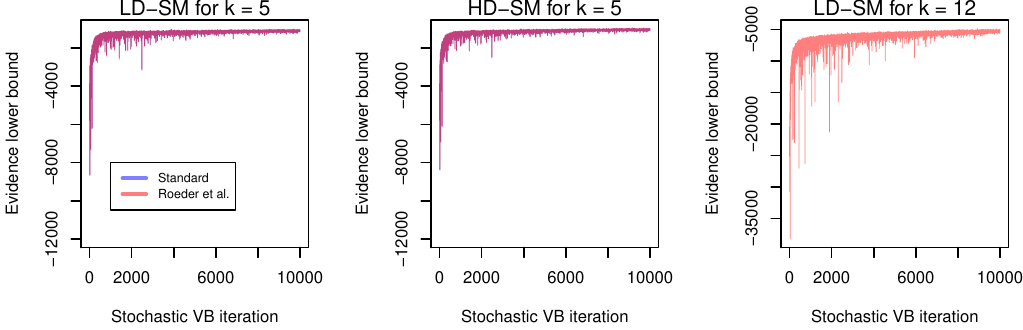}
\caption{${\cal L}(\lambda)$ for the variational approximations in the Wishart process example. The figure shows the estimated value of ${\cal L}(\lambda)$ vs iteration number using a low-dimensional state mean / high-dimensional state mean (LD-SM / HD-SM) with the gradient estimator in \citet{Roeder2017} or the standard estimator. The left and middle panels are
for $k=5$; the right panel is for the real data with $k=12$.}
\label{fig:ELBOs_PhilipovGlickman}
\end{figure}

\begin{figure}[h!]
\centering
\includegraphics[width=0.9\columnwidth]{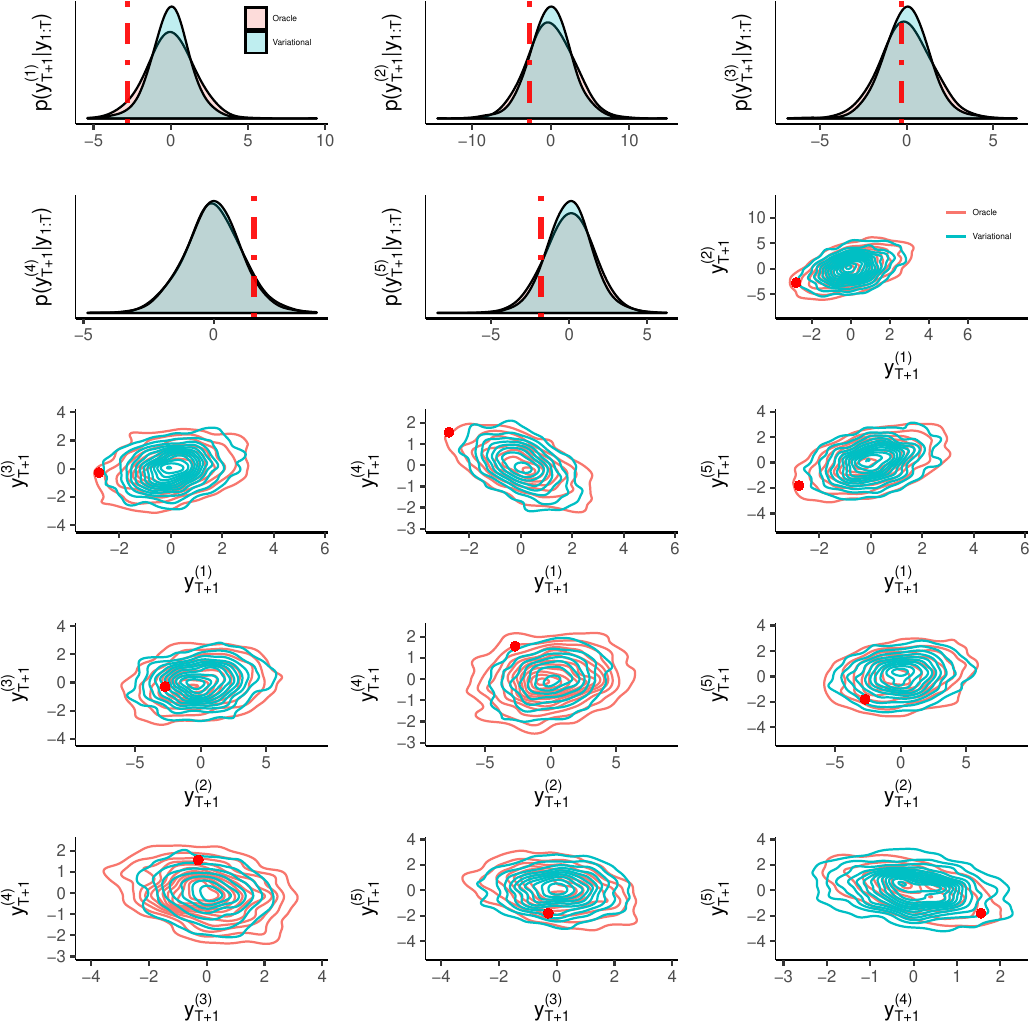}
\caption{Multivariate stochastic volatility model with simulated data and $T=100$. The top row  and the two panels from the
left  of the
second row  show the marginal one-step-ahead kernel density estimates of the predictive density for each of the $k=5$ variables
for both the variational approximation and the oracle; the test observation is the red line.
The right panel of the second row and the rest of the panels
 show the contour plots of the kernel density estimates of the one-step-ahead bivariate predictive densities
for the variational approximation and the oracle; the red dot is the test observation.}
\label{fig:OracleVsVariational}
\end{figure}

\subsection{Results for simulated data\label{subsec:Results_PhilipovGlickman}}
We now assess  the variational approximation  by comparing its out-of-sample predictive properties
against the oracle predictive density.  See  Section~\ref{subsec:data_PhilipovGlickman}.  The comparison is based
on data generated by the multivariate stochastic volatility model with $d=0.2$, $\nu = 20$ and  $A$ generated from
$ \sim \mathrm{Inv\textnormal{-}Wishart}(5, \mathrm{diag}(5))$.
While  the reported results are for a particular simulated dataset due to space restrictions,
we have verified that the same performance is obtained when the random number seed is changed and
$d$ and $\nu$ are varied. Figure~\ref{fig:OracleVsVariational} shows the kernel density estimates
for the marginals of all five
parameters and bivariate kernel density estimates for all pairs of variables for the
predictive $p(y_{T+1}|y_{1:T})$ (variational
and oracle) for $T= 100$.
The figure also shows the test observation (withheld when estimating the variational predictive and the oracle predictive).
Figure \ref{fig:OracleVsVariational_boxplots} shows boxplots of draws from all marginals of the
predictive densities $p(y_{T+h}|y_{1:T + h - 1})$  (variational and oracle) for the horizons $h = 1, 2, 3, 4$.
This figure also shows the withheld test observation which is within the prediction intervals of both methods.
Figure \ref{fig:OracleVsVariational_portfolios} shows, for each of the $H=4$ horizons, future predictions
(variational and oracle) of an equally weighted portfolio $w_{T+h} = \sum_{k=1}^5 (1/5)y_{(T+h)k}$ conditional
on the posteriors using the data $y_{1:(T + h -1)}$.
Section~\ref{subsec:VariationalPredictiveEvaluationsAppendix} of the supplement
gives more  plots that further confirm the accuracy of the variational predictive densities.

\begin{figure}[h!]
\centering
\includegraphics[width=0.9\columnwidth]{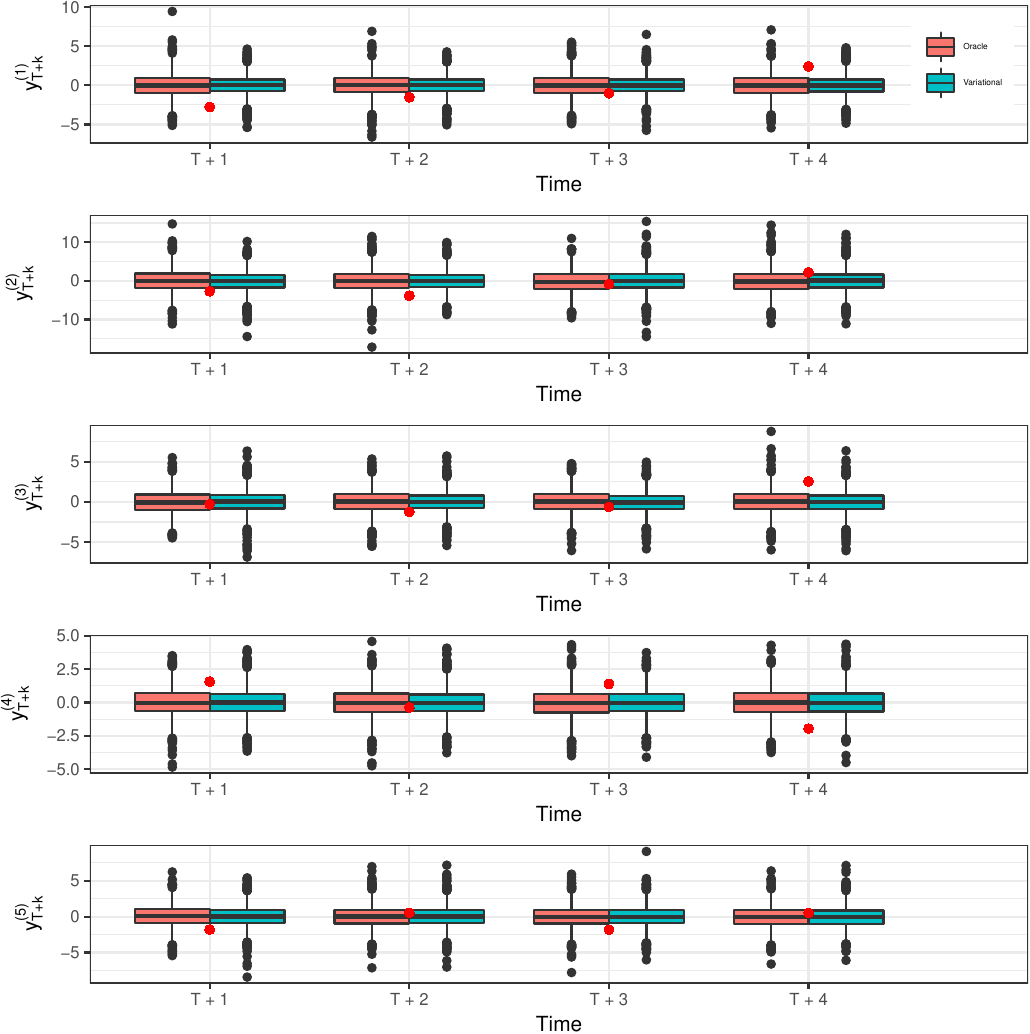}
\caption{Simulated data. Boxplots of samples from the variational one-step ahead marginal predictive densities compared against the oracle predictive densities  with $T=100,101,102,103$. The figure also shows the test observation (red) dot for each $T$ and variable.}
\label{fig:OracleVsVariational_boxplots}
\end{figure}

\begin{figure}[h!]
\centering
\includegraphics[width=0.9\columnwidth]{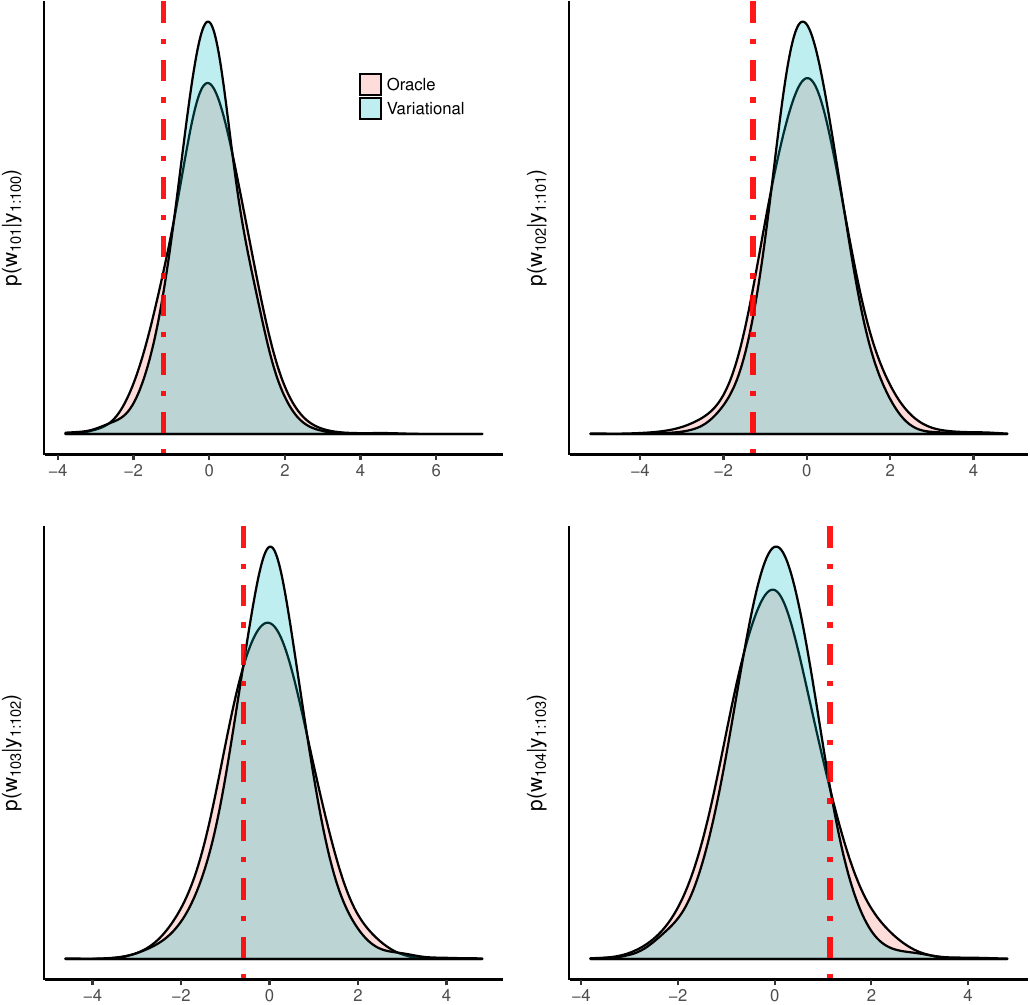}
\caption{Simulated data. Kernel density estimates of the one-step ahead predictive densities of a equally weighted
 portfolio  of assets. The results are for $T = 100, 101, 102, 103$. The figure also shows the test observation (red line) for each $T$.}
\label{fig:OracleVsVariational_portfolios}
\end{figure}

\subsection{Real data results\label{subsec:Results_PhilipovGlickman}}
The data consists of $T=100$ monthly observations on all $k=12$ (\citet{philipov+g06} only consider $k = 5$ assets and report an acceptance probability close to zero when $k=12$ for their sampler) value-weighted portfolios from the 201709 CRSP database, for the period 2009-06 to 2017-09. The portfolios are: consumer non-durables, consumer durables, manufacturing, energy, chemicals, business equipment, telecom, utilities, retail/wholesale, health care, finance, other. With $k=12$ the dimension of the state vector is  $p=78$. We follow \citet{philipov+g06} and prefilter each series using an AR(1) process.

The right panel in Figure~\ref{fig:ELBOs_PhilipovGlickman}  shows the estimated ELBO
on a variational optimization using the real dataset.
While the estimated ELBO plot  is more variable than for the $k=5$ case, it settles down eventually.
Figure~\ref{fig:Variational_real_data_k12} shows the in-sample prediction of $\tilde{y}_{100}$ given $y_{1:100}$, together with the observed data point,
for some of the assets. The figure also shows an in-sample prediction of a portfolio consisting of equally weighted assets. The variational posterior for the real data example uses the low-dimensional state mean parameterization.

\begin{figure}[h!]
\centering
\includegraphics[width=0.9\columnwidth]{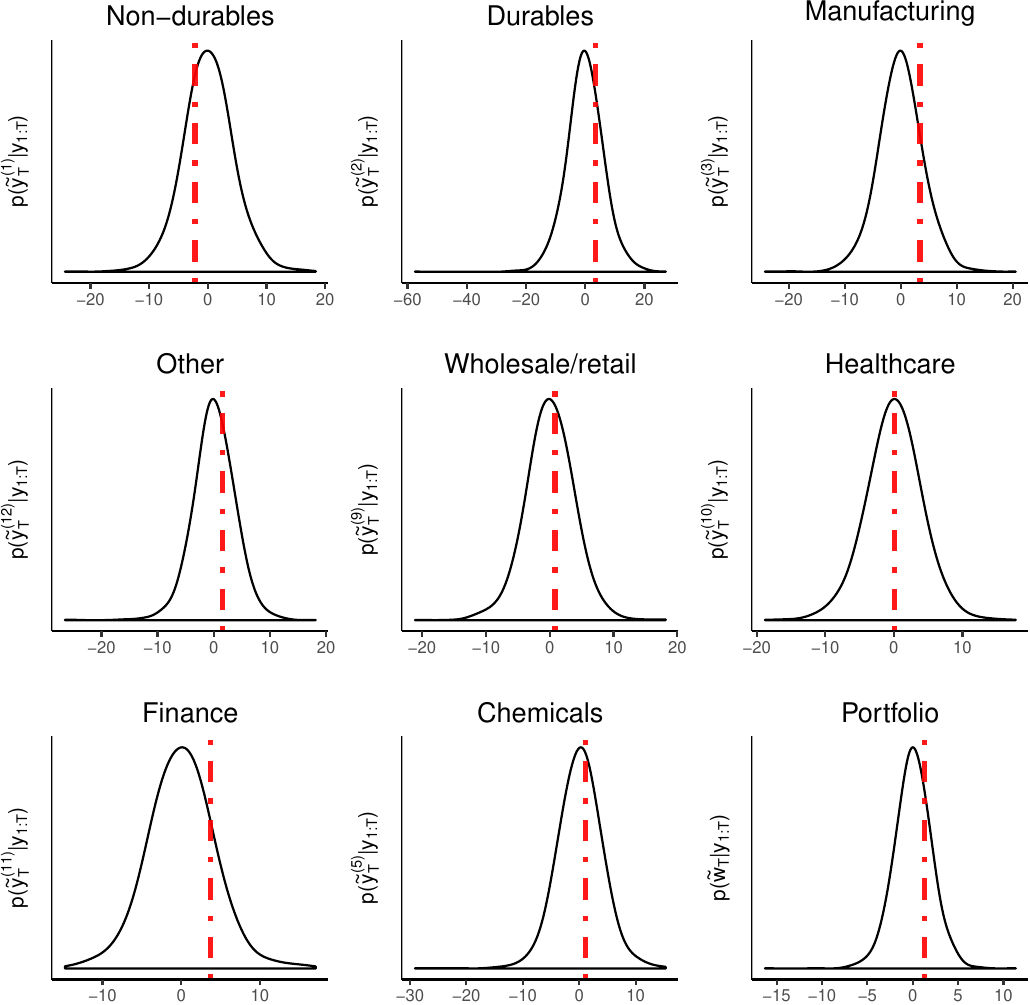}
\caption{Multivariate stochastic volatility model: Real data. Kernel density estimates of the in-sample predictive density for some the  assets and a equally weighted portfolio of  assets. The figure also shows the in-sample observation (red line).}
\label{fig:Variational_real_data_k12}
\end{figure}

\section{Discussion}

The article proposes  a  Gaussian variational approximation method
for high-dimensional state space models.
Dimension reduction in the variational approximation is achieved through a dynamic factor structure for the variational covariance matrix.
The factor structure reduces
the dimension in the description of the states, whereas the Markov dynamic structure for the factors achieves parsimony in describing the temporal
dependence.  We  show that the method works well in two challenging models. The first is an extended example for a spatio-temporal data set describing
the spread of the Eurasian collared-dove throughout North America. The second is a multivariate stochastic volatility model in
which the state vector is high dimensional.

Perhaps the most obvious limitation of our current work is the restriction to a Gaussian approximation, which does not allow capturing skewness or heavy tails
in the posterior distribution.  However, Gaussian variational approximations can be used as building blocks for more complex approximations based
on normal mixtures, copulas or conditionally Gaussian families for
example \citep{Han2016,Miller2016,smith+ln19,tan+bn19} and these more complex variational families can overcome some of the limitations
of the simple Gaussian approximation.  We intend to consider this in future work.

\section*{Acknowledgements}

We thank Mevin Hooten for his help with the Eurasian collared-dove data.
We thank Linda Tan for her comments on an early version of this manuscript. Matias Quiroz and Robert Kohn were partially supported
by Australian Research Council Center of Excellence grant CE140100049. David Nott was supported by a Singapore Ministry of Education Academic Research Fund Tier 2 grant (MOE2016-T2-2-135).

%\pagebreak
%\bibliographystyle{apalike}
%\addcontentsline{toc}{section}{\refname}
%\bibliography{ref}

%MQ:
\newpage
\bibliographystyle{apalike}
\addcontentsline{toc}{section}{\refname}
\bibliography{ref}
%\end{document}

\setcounter{equation}{0} % Reset equation numbering
%\newpage
%\renewcommand{\thesscheme}{S\arabic{sscheme}}
%\renewcommand{\thealgorithm}{S\arabic{algorithm}}
%\renewcommand{\theremark}{S\arabic{remark}}
\renewcommand{\theequation}{A\arabic{equation}}
\renewcommand{\thelemma}{A\arabic{lemma}}
\renewcommand{\thedefinition}{A\arabic{definition}}
\renewcommand{\theremark}{A\arabic{remark}}

\begin{appendices}
%\section{Some Notation}
\section{Gradient expressions of the evidence lower bound}\label{app:grad_expressions_lemmas}
\subsection{Notation and definitions}
We consider any vector $x \in \mathbb{R}^n$ to be arranged as a column vector with $n$ elements, i.e. $x=(x_1,\dots,x_n)^\top$. Likewise, if $g$ is function whose output is vector valued, i.e.  $g(x) \in \mathbb{R}^m$, then $g(x)=(g_1(x),\dots, g_m(x))^\top$. For a matrix $A$, $\mathrm{vec}(A)$ is the vector obtained by stacking the columns of $A$ from left to right.

\begin{definition}
\begin{enumerate}[topsep=0pt, label={\emph{(\roman*)}}]
\item
Suppose that $g : \mathbb{R}^n \to \mathbb{R}$ is a scalar valued function of a vector valued argument $x$. Then $\nabla_x g $ is a column vector with
 $i$th element $\partial g/\partial x_i$.
\item
Suppose that $g : \mathbb{R}^n \to \mathbb{R}^m$ is a vector valued function of a vector valued argument $x$. Then $dg/dx$ is a $ m \times n $
matrix with $(i,j)$th element $\mbox{$\partial g_i$}/\mbox{$\partial x_j$}$.
\item
Suppose that $g : \mathbb{R}^{m\times n} \to \mathbb{R}$ is a scalar valued function of a $m \times n $ matrix  $A= (a_{ij})$. Then
$\nabla_A g$ is an $ m \times n $ matrix with $(i,j)$th element $\mbox{$\partial g$}/\mbox{$\partial a_{ij}$}$.
\item
Suppose that $G : \mathbb{R}^{m\times n} \to \mathbb{R}^{q\times r}$ is a matrix valued function of a matrix valued argument $A$. Then, $$\frac{dG}{dA}  \coloneqq \frac{d\mathrm{vec}(G)}{d\mathrm{vec}(A)}$$
is an $mq \times nr$ matrix with $(i,j)$th element $\mbox{$\partial \mathrm{vec}(G)_i$}/\mbox{$\partial \mathrm{vec}(A)_j$}$.
\end{enumerate}
\end{definition}

\begin{remark}
If $g$ is a scalar function of a vector valued argument $x$, then Part (ii) (with $m=1$) implies that $dg/dx$ is a row vector. Hence, $\nabla_X g = (dg/dx)^\top$.
\end{remark}
We write $0_{m\times n}$ for the $m\times n$ matrix of zeros, $\otimes$ for the
Kronecker product and $\odot$ for the Hadmard (elementwise) product which can be applied to two matrices of the same dimensions. For an $m\times n$ matrix $A$, $\mathrm{vec}(A)$ is the vector obtained by stacking the columns of $A$ from left to right. We also write $K_{r,s}$ for the commutation
matrix, of dimensions $rs\times rs$, which for an $r\times s$ matrix $Z$ satisfies
$$K_{r,s}\mathrm{vec}(Z)=\mathrm{vec}(Z^\top).$$

\subsection{Results\label{ss: results app}}
We adopt the notation from Section \ref{subsec:structure_variational_posterior} and
construct the variational distribution for $\theta$ through
\begin{align*}
  \theta & = \left[\begin{array}{cc} X \\ \zeta \end{array}\right] = \left[\begin{array}{cc} I_{T+1}\otimes B & 0 \\ 0 & I_P\end{array}\right]\rho + \left[\begin{array}{c} \epsilon \\ 0 \end{array}\right];
\end{align*}
$P = \dim(\zeta)$ and $\epsilon=(\epsilon_0^\top,\dots,\epsilon_T^\top)^\top$; $\epsilon_t$ is defined in \eqref{ldsm}.
Apply the reparameterization trick for the LD-SM parameterization
(see the discussion in Section~\ref{subsec:structure_variational_posterior}) and write
\begin{align}
\theta & = W\mu+WC^{-\top}\omega+Ze;  \label{reparformula2}
\end{align}
where $\omega\sim \mathcal{N}(0,I_{q(T+1)})$;
$$W=\left[\begin{array}{cc} I_{T+1}\otimes B & 0_{p(T+1)\times P} \\ 0_{P\times q(T+1)} & I_P \end{array}\right], \;\;\;\;Z=\left[\begin{array}{cc} D & 0_{p(T+1)\times P} \\ 0_{P\times p(T+1)} & 0_{P\times P} \end{array}\right],\;\;\;\; e=\left[\begin{array}{c} \epsilon \\ 0_{P\times 1} \end{array}\right]; $$
and $D$ is a diagonal matrix with diagonal entries $(\delta_0^\top,\dots,\delta_T^\top)^\top$.
Then, the distribution of $(\omega, \epsilon) \sim \mathcal{N}(0,I_{(p+q)(T+1)+P})$  does
 not depend on the variational parameter $\lambda$.

 \medskip
 Lemmas~\ref{lem:standard_gradient} and \ref{lem:Roeder_gradient} give the gradients of the ELBO in \eqref{lbqr} using the reparameterization trick. These can be used for unbiased gradient estimation of the lower bound by sampling one or more samples from $(\omega, \epsilon) $. Lemma \ref{lem:standard_gradient} (\ref{lem:Roeder_gradient}) contains the gradients corresponding to \eqref{lbqr2} (\eqref{lbq2}), which we refer to as the standard gradient (the \citeauthor{Roeder2017} gradient).
  By    the discussion in Section~\ref{sec:stochasticgradientvariational}, the \citeauthor{Roeder2017} gradient has the property that a Monte Carlo estimate of the gradient based on (\ref{sgradmuss}) using a single sample is zero when the variational posterior is exact.
The proofs of the lemmas are in Section~\ref{app:GradientVariationalApprox} of the supplement.

\begin{lemma}[Standard gradient]\label{lem:standard_gradient}
Let $\mathcal{L}(\lambda)=
E_{(\omega, \epsilon)  }\left(\log h(\theta) - \log q_\lambda(\theta) \right),
$
with $\theta$ in \eqref{reparformula2}, $(\omega, \epsilon)$ as above,
 and $q_\lambda(\theta) = \mathcal{N}\left(\theta|W\mu, W \Sigma W^\top + Z^2\right).$ If  $h(\theta)$ is differentiable,
 then,
\begin{enumerate}[topsep=0pt, label={\emph{(\roman*)}}]
\item \begin{align}
  \nabla_\mu {\cal L}(\lambda) & = W^\top E_{(\omega, \epsilon) }(\nabla_\theta \log h(W\mu+WC^{-\top}\omega+Ze));  \label{gradmuss}
\end{align}
\item \begin{align}
  \nabla_{\mathrm{vec}(B)} {\cal L}(\lambda) & = T_{1B}+T_{2B}+T_{3B}, \label{gradBss}
\end{align}
where
\begin{align*}
T_{1B} & = \left\{\frac{dW}{dB}\right\}^\top E_{(\omega, \epsilon) }(((\mu+C^{-\top}\omega)\otimes I_{p(T+1)+P}) \nabla_\theta \log h(W\mu+WC^{-\top}\omega+Ze)),\\
  T_{2B} & = \left\{\frac{dW}{dB}\right\}^\top \mathrm{vec}((W\Sigma W^\top+Z^2)^{-1}W\Sigma ),\\
 T_{3B} & = \left\{\frac{dW}{dB}\right\}^\top E_{(\omega, \epsilon)}\left(\mathrm{vec}\left((W\Sigma W^\top+Z^2)^{-1}(WC^{-\top}\omega+Ze)\omega^\top C^{-1}\right.\right. \\
 & \left.\left.-(W\Sigma W^\top+Z^2)^{-1}(WC^{-\top}\omega+Ze)(WC^{-\top}\omega+Ze)^\top (W\Sigma W^\top+Z^2)^{-1}W\Sigma\right)\right);
\end{align*}
above,
\begin{align*}
   \frac{dW}{dB} = & (Q_1^\top\otimes P_1)\left[\left\{(I_{T+1}\otimes K_{q,(T+1)})(\mbox{$\mathrm{vec}(I_{T+1})$}\otimes I_q)\right\}\otimes I_p\right],
\end{align*}
with
$$P_1=\left[\begin{array}{c} I_{p(T+1)} \\ 0_{P\times p(T+1)} \end{array}\right],\;\;\;\;
Q_1=\left[\begin{array}{cc} I_{q(T+1)} & 0_{q(T+1)\times P} \end{array}\right];$$

\item \begin{align}
 \nabla_\delta {\cal L}(\lambda) & = E_{(\omega, \epsilon)}(\mathrm{diag}(\nabla_X \log h(W\mu+WC^{-\top}\omega+Ze)\epsilon^\top+(W_1\Sigma_1 W_1^\top+D^2) ^{-1}D \nonumber \\
 &  +(W_1\Sigma_1 W_1^\top+D^2)^{-1}(W_1 C_1^{-\top}\omega_1+D\epsilon)\epsilon^\top \nonumber \\
 & -(W_1\Sigma_1 W_1^\top+D^2)^{-1}(W_1C_1^{-\top}\omega_1+D\epsilon)(W_1C_1^{-\top}\omega_1+D\epsilon)^\top (W_1\Sigma_1W_1^\top+D^2)^{-1}D)), \label{graddss}
\end{align}
where $W_1=I_{T+1}\otimes B$;
\item \begin{align}
\nabla_C {\cal L}(\lambda) = & E_{(\omega, \epsilon) }\left( -C^{-\top}\omega \nabla_\theta \log h(W\mu+WC^{-\top}\omega+Ze)^\top WC^{-\top}\right. \nonumber \\
  - & \Sigma W^\top (W\Sigma W^\top+Z^2)^{-1} WC^{-\top} \nonumber \\
  - & C^{-\top}\omega(WC^{-\top}\omega+Ze)^\top (W\Sigma W^\top+Z^2)^{-1}WC^{-\top} \nonumber \\
  + & \left. \Sigma W^{\top}(W\Sigma W^\top+Z^2)^{-1}(WC^{-\top}\omega+Ze)(WC^{-\top}\omega+Ze)^\top (W\Sigma W^\top+Z^2)^{-1}WC^{-\top}\right) \label{gradCss}
\end{align}
\end{enumerate}
\end{lemma}

\begin{lemma}[\citeauthor{Roeder2017} gradient]\label{lem:Roeder_gradient}
Let $\mathcal{L}(\lambda)=E_{(\omega, \epsilon) }\left(\log h(\theta) - \log q_\lambda(\theta) \right)$,
with $\theta$ defined in \eqref{reparformula2}, ${(\omega, \epsilon) }$ as above and $q_\lambda(\theta) = \mathcal{N}\left(\theta|W\mu, W \Sigma W^\top + Z^2\right).$

If $h(\theta)$ is differentiable, then
\begin{enumerate}[topsep=0pt, label={\emph{(\roman*)}}]
\item \begin{align}
  \nabla_\mu {\cal L}(\lambda) & = W^\top E_{(\omega, \epsilon) }(\nabla_\theta \log h(W\mu+WC^{-\top}\omega+Ze)+(W\Sigma W^\top+Z^2)^{-1}(WC^{-\top}\omega+Ze));  \label{sgradmuss}
\end{align}
\item \begin{align}
  \nabla_{\mathrm{vec}(B)} {\cal L}(\lambda) & = T_{1B}+T_{3B}', \label{sgradBss}
\end{align}
with $T_{1B}$ as in Part (ii) of Lemma \ref{lem:standard_gradient} and
\begin{align}
 T_{3B}' & = \left\{\frac{dW}{dB}\right\}^\top E_{(\omega, \epsilon) }\left(\mathrm{vec}\left((W\Sigma W^\top+Z^2)^{-1}(WC^{-\top}\omega+Ze)(\omega^\top C^{-1}+\mu^\top)\right)\right); \label{T3Bp}
\end{align}
\item \begin{align}
 \nabla_\delta {\cal L}(\lambda) & = E_{(\omega, \epsilon) }(\mathrm{diag}(\nabla_X \log h(W\mu+WC^{-\top}\omega+Ze)\epsilon^\top  \nonumber \\
 & +(W_1\Sigma_1 W_1^\top+D^2)^{-1}(W_1 C_1^{-\top}\omega_1+D\epsilon)\epsilon^\top)), \label{sgraddss}
\end{align}
where $W_1=I_{T+1}\otimes B$;
\item \begin{align}
\nabla_C {\cal L}(\lambda) = & E_{(\omega, \epsilon) }\left( -C^{-\top}\omega \left\{ \nabla_\theta \log h(W\mu+WC^{-\top}\omega+Ze)^\top \right. \right. \nonumber \\
+ & \left. \left. (WC^{-\top}\omega+Ze)^\top (W\Sigma W^\top+Z^2)^{-1} \right\} WC^{-\top}\right).
 \label{sgradCss}
\end{align}
\end{enumerate}

\end{lemma}

\end{appendices}

\newpage

%%%%%%%%%%%% BELOW IS THE SUPPLEMENT:
%\newpage
%\renewcommand{\thesscheme}{S\arabic{sscheme}}
%\renewcommand{\thealgorithm}{S\arabic{algorithm}}
%\renewcommand{\theremark}{S\arabic{remark}}
\renewcommand{\theequation}{S\arabic{equation}}
\renewcommand{\thesection}{S\arabic{section}}
\renewcommand{\thepage}{S\arabic{page}}
\renewcommand{\thetable}{S\arabic{table}}
\renewcommand{\thefigure}{S\arabic{figure}}
\renewcommand{\thedefinition}{S\arabic{definition}}
\renewcommand{\theremark}{S\arabic{remark}}
\setcounter{page}{1}
\setcounter{section}{0}
\setcounter{equation}{0}
\setcounter{table}{0}

Online supplement to \lq Gaussian variational approximations for high-dimensional state-space models\rq.

We refer to equations, sections, etc in the main paper as (1), Section 1, etc, and in the supplement as (S1), Section S1, etc.

% David had: \bibliographystyle{chicago}
%\bibliography{ref}

%MQ:
%\bibliographystyle{apalike}
%\addcontentsline{toc}{section}{\refname}\bibliography{ref}

%\appendix
\section{Notation and definitions}\label{supp:notation}
To make the supplement easier to follow we repeat some of the material
in Appendix~\ref{app:grad_expressions_lemmas} and add to it the new notation needed in this supplement.
We consider any $n$ dimensional $x$ to be  a column vector, i.e. $x=(x_1,\dots,x_n)^\top$; thus, if $g$ is function whose output is  $m$ dimensional, then $g(x)=(g_1(x),\dots, g_m(x))^\top$.
For a matrix $A$, $\mathrm{vec}(A)$ is the vector obtained by stacking the columns of $A$ from left to right.
\begin{definition}
\begin{enumerate}[topsep=0pt, label={\emph{(\roman*)}}]
\item
Suppose that $g : \mathbb{R}^n \to \mathbb{R}$ is a scalar valued function of a vector valued argument $x$. Then $\nabla_x g $ is a column vector with
 $i$th element $\partial g/\partial x_i$.
\item
Suppose that $g : \mathbb{R}^n \to \mathbb{R}^m$ is a vector valued function of a vector valued argument $x$. Then $dg/dx$ is a $ m \times n $
matrix with $(i,j)$th element $\mbox{$\partial g_i$}/\mbox{$\partial x_j$}$.
\item
Suppose that $g : \mathbb{R}^{m\times n} \to \mathbb{R}$ is a scalar valued function of a $m \times n $ matrix  $A= (a_{ij})$. Then
$\nabla_A g$ is an $ m \times n $ matrix with $(i,j)$th element $\mbox{$\partial g$}/\mbox{$\partial a_{ij}$}$.
\item
Suppose that $G : \mathbb{R}^{m\times n} \to \mathbb{R}^{q\times r}$ is a matrix valued function of a matrix valued argument $A$. Then, $$\frac{dG}{dA}  \coloneqq \frac{d\mathrm{vec}(G)}{d\mathrm{vec}(A)}$$
is an $mq \times nr$ matrix with $(i,j)$th element $\mbox{$\partial \mathrm{vec}(G)_i$}/\mbox{$\partial \mathrm{vec}(A)_j$}$.
\end{enumerate}
\end{definition}
\begin{remark}
If $g$ is a scalar function of a vector valued argument $x$, then Part (ii) (with $m=1$) implies that $dg/dx$ is a row vector. Hence, $\nabla_X g = (dg/dx)^\top$.
\end{remark}
\begin{remark}
Part (iv), with $r = 1$, covers the case of the derivative of a vector valued function with respect to a matrix valued argument.
\end{remark}
Let $0_{m\times n}$ be the $m\times n$ matrix of zeros; let $\otimes$ be the
Kronecker product and $\odot$ the Hadmard (elementwise) product, both of which can be applied to two matrices
of the same dimensions. For a $n \times n $ symmetric matrix $A$,
$\mathrm{vech}(A)$ is the column vector of length $n(n + 1)/2$ obtained by vectorizing the lower triangular and diagonal parts
 of $A$;
the  $L_n$ elimination matrix is defined by $\mathrm{vech}(A)=:L_k \mathrm{vec}(A)$; and the duplication matrix $D_n$ is defined by $\mathrm{vec}(A)=:D_k\mathrm{vech}(A)$;
see \cite{magnus1980elimination} for further properties of the elimination and duplication matrices.
We also write $K_{r,s}$ for the commutation matrix of dimensions $rs\times rs$, which for an $r\times s$ matrix $Z$ satisfies $$K_{r,s}\mathrm{vec}(Z)=\mathrm{vec}(Z^\top).$$

\section{Details on the sparsity of the precision matrix \label{app:SparsityPrecisionMatrix} of the dynamic factors}
Write $\Sigma_1$ for the covariance matrix of the latent dynamic factors $z=(z_0^\top,\dots, z_T^\top)^\top$ and let $\Omega_1 = \Sigma^{-1}_1$ be the corresponding precision matrix. Denote by $C_1$ the lower triangular Cholesky factor of $\Omega_1$, i.e. $\Omega_1 = C_1C^\top_1$. Section \ref{subsec:structure_variational_posterior} assumes that
$C_1$ takes the form
\begin{align}
C_1 &= \begin{bmatrix}
C_{00}        & 0              & 0                  & \ldots        & 0                    & 0                    \\
C_{10}        & C_{11}     & 0                  & \ldots        & 0                    & 0                    \\
0                 & C_{21}      & C_{22}        & \ldots        & 0                     & 0                   \\
\vdots         & \vdots       & \vdots         & \ddots       & \vdots            & \vdots            \\
0                 & 0                & 0                 & \ldots        & C_{T-1,T-1}   & 0                   \\
0                 &  0               & \ldots          & \ldots        & C_{T,T-1}      & C_{TT}
\end{bmatrix},
\end{align}
where the blocks in this block partitioned matrix follow the blocks of $z$;
the corresponding precision matrix takes the form
\begin{align}
\Omega_1 &= \begin{bmatrix}
\Omega_{00} & \Omega_{10}^\top  & 0 & \ldots & 0 & 0  \\
\Omega_{10} & \Omega_{11} & \Omega_{21}^\top  & \ldots & 0 & 0  \\
0 & \Omega_{21} & \Omega_{22}  & \ldots & 0 & 0 \\
\vdots & \vdots & \vdots & \ddots & \ldots & \vdots & \\
0 & 0 & 0  &  \ldots  & \Omega_{T-1,T-1} & \Omega_{T,T-1}^\top \\
0 &  0 & 0  & \ldots & \Omega_{T,T-1} & \Omega_{TT} \\
\end{bmatrix}.
\end{align}

\section{Derivations\label{app:AppendixA}}
\label{sec:Gradient_expressions}
\subsection{Gradients of the variational approximation\label{app:GradientVariationalApprox}}
This section proves Lemmas \ref{lem:standard_gradient} and \ref{lem:Roeder_gradient}, which contain the gradient with respect to the variational parameters when applying the reparameterization trick as outlined in Section \ref{subsec:structure_variational_posterior}. The following result about the $\mathrm{vec}$ and Kronecker product is useful.  For conformably dimensioned matrices $A$, $B$ and $C$,
\begin{align*}
  \mathrm{vec}(ABC) & =(C^\top \otimes A)\mathrm{vec}(B).
\end{align*}

Using the notation in Section \ref{sec:Gaussian_Variational_high_dim_state} and Appendix \ref{app:grad_expressions_lemmas},
$\theta\sim q_\lambda(\theta)$, and its generative form is given by $\theta=W\mu+WC^{-\top}\omega+Ze$; thus,
\begin{align}
 \log q_\lambda(\theta) = & -\frac{p(T+1)+P}{2}\log 2\pi-\frac{1}{2}\log |W\Sigma W^\top+Z^2|  \nonumber \\
 & -\frac{1}{2}(WC^{-\top}\omega+Ze)^\top (W\Sigma W^\top+Z^2)^{-1}(WC^{-\top}\omega+Ze). \label{logq}
\end{align}
\begin{proof}[Proof of Lemma \ref{lem:standard_gradient}]
Proof of Part (i). Since \eqref{logq} does not depend on $\mu$,
\begin{align*}
 \nabla_\mu {\cal L}(\lambda) = & \nabla_\mu E_{ ( \omega, \epsilon )}(\log h(W\mu+WC^{-\top}\omega+Ze)) \\
 = & E_{ ( \omega, \epsilon )}(W^\top \nabla_\theta \log h(W\mu+WC^{-\top}\omega+Ze)).
\end{align*}

Proof of Part (ii). For the parametrization outlined in Section \ref{subsec:OngEtAl_review}, we use the derivations in \citet{ong+ns17},
\begin{align}
 \nabla_{\mbox{vec$(B)$}} E_{ ( \omega, \epsilon )}\left(\frac{1}{2}\log |BB^\top+D^2|\right)= &
\mathrm{vec}((BB^\top+D^2)^{-1} B),  \label{gradexp1}
\end{align}
{\small
\begin{align}
\nabla_B E_{ ( \omega, \epsilon )}\Bigl(-\frac{1}{2}\mbox{tr}((B\omega+\delta\odot\kappa)^\top & (BB^\top+D^2)^{-1}(B\omega+\delta\odot\kappa))\Bigr) =
 E_{ ( \omega, \epsilon )}(-(BB^\top+D^2)^{-1}(B\omega+\delta\odot\kappa)\omega^\top \nonumber \\
  & +(BB^\top+D^2)^{-1}(B\omega+\delta\odot\kappa)(B\delta+\delta\odot\kappa)^\top (BB^\top+D^2)^{-1}B), \label{gradexp2}
\end{align}}
\begin{align}
  \nabla_\delta {\cal L}(\lambda) = & E_{ ( \omega, \epsilon )}(\mbox{diag}(\nabla_\theta h(\mu+B\omega+\delta\odot\kappa)\kappa^\top+(BB^\top+D^2)^{-1}(B\omega+\delta\odot\kappa)\kappa^\top).  \label{gradd}
\end{align}
In deriving an expression for $\nabla_{\mathrm{vec}(B)} {\cal L}(\lambda)$, it is helpful to have an explicit expression for
$dW/dB$.  We can write $W=W_1+W_2$, with $\mathrm{vec}(W)=\mathrm{vec}(W_1)+\mathrm{vec}(W_2)$; and
$$W_1=\left[\begin{array}{cc} I_{T+1}\otimes B & 0_{p(T+1)\times P} \\ 0_{P\times q(T+1)} & 0_{P\times P} \end{array}\right],\;\;\;\; W_2=\left[\begin{array}{cc}
0_{p(T+1)\times q(T+1)} & 0_{p(T+1)\times P} \\ 0_{P\times q(T+1)} & I_P \end{array}\right].$$
Using Theorem 1 of \citet{caswell+v16},
$\mathrm{vec}(W_1)=(Q_1^\top\otimes P_1)\mathrm{vec}(I_{T+1}\otimes B)$;
with
$$P_1=\left[\begin{array}{c} I_{p(T+1)} \\ 0_{P\times p(T+1)} \end{array}\right],\;\;\;\;
Q_1=\left[\begin{array}{cc} I_{q(T+1)} & 0_{q(T+1)\times P} \end{array}\right];$$
$\mathrm{vec}(W_2)$ can be written similarly, but its expression is unnecessary  since
$W_2$ does not depend on $B$.  Using standard results on the differentiation of Kronecker products,
\begin{align}
  \frac{dW}{dB} = & (Q_1^\top\otimes P_1) \frac{d(I_{T+1}\otimes B)}{dB},\\
%\end{align}
%which, using standard results concerning differentiation of Kronecker products, gives
%\begin{align}
 = & (Q_1^\top\otimes P_1)\left[\left\{(I_{T+1}\otimes K_{q,(T+1)})(\mbox{$\mathrm{vec}(I_{T+1})$}\otimes I_q)\right\}\otimes I_p\right].
\label{dvecWdvecB}
\end{align}
Then,
\begin{align*}
\nabla_{\mathrm{vec}(B)} {\cal L}(\lambda) = & T_{1B}+T_{2B}+T_{3B},
\end{align*}
where
\begin{align}
  T_{1B} = & E_{ ( \omega, \epsilon )}(\nabla_{\mathrm{vec}(B)} \log h(W\mu+W C^{-\top} \omega+Ze)), \nonumber \\
 & = E_{ ( \omega, \epsilon )}\left(\left\{\frac{dW}{dB}\right\}^\top((\mu+C^{-\top}\omega)\otimes I_{p(T+1)+P}) \nabla_\theta \log h(W\mu+W C^{-\top}\omega+Ze)\right); \label{T1B}
\end{align}
\begin{align}
  T_{2B} = & \left\{\frac{dW}{dB}\right\}^\top \left\{\frac{dWC^{-\top}}{dW}\right\}^\top \nabla_{\mathrm{vec}(WC^{-\top})} \left\{\frac{1}{2}\log |WC^{-T}C^{-1}W^\top + Z^2|\right\} \nonumber \\
  = & \left\{\frac{dW}{dB}\right\}^\top (C^{-\top}\otimes I_{p(T+1)+P}) \mathrm{vec}((WC^{-\top}C^{-1}W^\top+Z^2)^{-1}WC^{-\top}) \nonumber \\
 = & \left\{\frac{dW}{dB}\right\}^\top
\mathrm{vec}((WC^{-\top}C^{-1}W^\top+Z^2)^{-1}WC^{-\top} C^{-1}), \label{T2B}
\end{align}
using (\ref{gradexp1}); and using (\ref{gradexp2})
{\small
\begin{align}
  T_{3B} = & \nabla_{\mathrm{vec}(B)}E_{ ( \omega, \epsilon )}\left(\frac{1}{2}(WC^{-\top}\omega+Ze)^\top (W\Sigma W^\top+Z^2)^{-1}(WC^{-\top}\omega+Ze) \right) \nonumber \\
 = &  \left\{\frac{dW}{dB}\right\}^\top  \left\{\frac{dWC^{-\top}}{dW}\right\}^\top
\nabla_{\mathrm{vec}(WC^{-\top})}E_f\left(\frac{1}{2}(WC^{-\top}\omega+Ze)^\top (W\Sigma W^\top+Z^2)^{-1} \right.  \nonumber \\
~ & \left. (WC^{-\top}\omega+Ze) \right) \nonumber \\
 = & \left\{\frac{dW}{dB}\right\}^\top (C^{-\top}\otimes I_{p(T+1)+P}) \mathrm{vec}\left(E_f\left\{(W\Sigma W^\top+Z^2)^{-1}(WC^{-\top}\omega+Ze)\omega^\top \right.\right. \nonumber \\
 - & \left. \left. (W\Sigma W^\top+Z^2)^{-1}(WC^{-\top}\omega+Ze)(WC^{-\top}\omega+Ze)^\top (W\Sigma W^\top+Z^2)^{-1}WC^{-\top}\right\} \right). \label{T3B}
\end{align}}
Equation~(\ref{gradBss}) is obtained by combining (\ref{T1B})-- (\ref{T3B}).

Proof of Part (iii). The derivation of the gradient is identical to that of (\ref{gradd}), giving (\ref{graddss}) directly.

Proof of Part (iv). Writing
 \begin{align*}
\nabla_{\mathrm{vec}(C)} {\cal L}(\lambda) = & T_{1C}+T_{2C}+T_{3C},
\end{align*}
where
{\small
\begin{align}
 T_{1C} = & \nabla_{\mathrm{vec}(C)}E_{ ( \omega, \epsilon )}( \log h(W\mu+WC^{-\top}\omega+Ze)) \nonumber \\
 = &  \left\{\frac{dC^{-1}}{dC}\right\}^\top  \left\{\frac{dC^{-\top}}{dC^{-1}}\right\}^\top
  \left\{\frac{dWC^{-\top}\omega}{dWC^{-\top}}\right\}^\top E_f(\nabla_\theta \log  h(W\mu+WC^{-\top}\omega+Ze)) \nonumber \\
 = & E_{ ( \omega, \epsilon )}\left(-(C^{-1}\otimes C^{-\top})K_{q(T+1)+P,q(T+1)+P}(I_{q(T+1)+P}\otimes W^\top)(\omega\otimes I_{q(T+1)+P}) \right. \nonumber \\
 ~ & \left. \nabla_\theta \log  h(W\mu+WC^{-\top}\omega+Ze)\right) \nonumber \\
 = & -E_{ ( \omega, \epsilon )}(\mathrm{vec}(C^{-\top}\omega \nabla_\theta  \log h(W\mu+WC^{-\top}\omega+Ze)^\top W C^{-\top})); \label{T1C}
\end{align}}
\begin{align}
 T_{2C} = & \nabla_{\mathrm{vec}(C)} \frac{1}{2} \log |W C^{-\top} C^{-1}W^\top + Z^2| \nonumber \\
 = & \left\{\frac{dC^{-1}}{dC}\right\}^\top  \left\{\frac{dC^{-\top}}{dC^{-1}}\right\}^\top \left\{\frac{dWC^{-\top}}{dC^{-\top}}\right\}^\top \nabla_{\mathrm{vec}(WC^{-\top})}  \frac{1}{2}\log |WC^{-\top}C^{-1}W^\top+Z^2| \nonumber \\
 = & -(C^{-1}\otimes C^{-\top})K_{q(T+1)+P,q(T+1)+P}^\top (I_{q(T+1)+P}\otimes W^\top )\mathrm{vec}((WC^{-\top}C^{-1} W^\top +Z^2)^{-1}WC^{-\top} \nonumber \\
 = & -(C^{-1}\otimes C^{-\top})K_{q(T+1)+P,q(T+1)+P}^\top \mathrm{vec}(W^\top (WC^{-\top}C^{-1}W^\top+Z^2)^{-1}WC^{-\top}) \nonumber \\
 = & -(C^{-1}\otimes C^{-\top})\mathrm{vec}(C^{-1}W^\top (WC^{-\top}C^{-1}W^\top + Z^2)^{-1} W) \nonumber \\
 = & -\mathrm{vec}(C^{-\top}C^{-1}W^\top (WC^{-\top}C^{-1}W^\top+Z^2)^{-1}WC^{-\top}); \label{T2C}
\end{align}
and $T_{3C}=E_{(\omega, \epsilon)}(R_{3C})$, where
{\small
\begin{align}
 R_{3C} = & \nabla_{\mathrm{vec}(C)} \frac{1}{2}(WC^{-\top}\omega+Ze)^\top (WC^{-T}C^{-1}W^\top+Z^2)^{-1}(WC^{-\top}\omega+Ze) \nonumber \\
 = & \left\{\frac{dC^{-1}}{dC}\right\}^\top  \left\{\frac{dC^{-\top}}{dC^{-1}}\right\}^\top \left\{-\frac{dWC^{-\top}}{dC^{-\top}}\right\}^\top  \nonumber \\
 & \nabla_{\mathrm{vec}(WC^{-\top})} \frac{1}{2} (WC^{-\top}\omega+Ze)^\top (WC^{-T}C^{-1}W^\top+Z^2)^{-1}(WC^{-\top}\omega+Ze) \nonumber \\
 = & -(C^{-1}\otimes C^{-\top})K_{q(T+1)+P,q(T+1)+P}^\top (I_{q(T+1)+P}\otimes W^\top) \mathrm{vec}\left((WC^{-\top}C^{-1}W^\top+Z^2)^{-1} \right. \nonumber \\ ~ &  \left. (WC^{-\top}+D\epsilon)\omega^\top -(WC^{-\top}C^{-1}W^\top+Z^2)^{-1}(WC^{-\top}\omega + Ze)(WC^{-\top}\omega+Z e)^\top \right. \nonumber \\
~ & \left. (WC^{-\top}C^{-1}W^\top+Z^2)^{-1}WC^{-\top}\right) \nonumber \\
 = & -(C^{-1}\otimes C^{-\top})K_{q(T+1)+P,q(T+1)+P} \mathrm{vec}\left(W^\top(W\Sigma^{-1}W^\top+Z^2)^{-1}(WC^{-\top}\omega+Ze)\omega^\top\right. \nonumber \\
 & \left.-W^\top(W\Sigma^{-1}W^\top+Z^2)^{-1}(WC^{-\top}\omega+Ze)(WC^{-\top}\omega+Ze)^\top (WC^{-\top}C^{-1} W^\top+Z^2)^{-1}WC^{-\top}\right) \nonumber \\
 = & -(C^{-1}\otimes C^{-\top})\mathrm{vec}\left(\omega (WC^{-\top}\omega+Ze)^\top(W\Sigma W^\top+Z^2)^{-1}W \right. \nonumber \\
 & \left.-C^{-1}W^\top (W\Sigma W^\top+Z^2)^{-1}(WC^{-\top}\omega+Ze)(WC^{-\top}\omega+Ze)^\top (W\Sigma W^\top+Z^2)^{-1}WC^{-\top}\right) \nonumber \\
 = & \mathrm{vec}\left(-C^{-\top}\omega (WC^{-\top}\omega+Ze)^\top (W\Sigma W^\top+Z^2)^{-1}WC^{-\top}\right. \nonumber \\
 & \left.+C^{-\top}C^{-1}W^\top (WC^{-\top}C^{-1}W^\top+Z^2)^{-1}(WC^{-\top}\omega+Ze)(WC^{-\top}\omega+Ze)^\top \right.  \nonumber \\ ~ &  \left. (W\Sigma W^\top+Z^2)^{-1}WC^{-\top}\right). \label{T3C}
\end{align}}
Equation (\ref{gradCss}) is obtained by combining (\ref{T1C})-- (\ref{T3C}).
\end{proof}
\begin{proof}[Proof of Lemma \ref{lem:Roeder_gradient}]
Proof of Part (i). The expectation of the second term in \eqref{sgradmuss} is zero and thus the expression becomes \eqref{gradmuss}.

Proof of Part (ii). The term $T_{2B}$ in \eqref{lem:standard_gradient} cancels the second term in $T_{3B}$ in \eqref{lem:standard_gradient} after taking expectations, leaving \eqref{sgradBss}.

Proof of Part (iii). Cancellation of the second and fourth terms in \eqref{graddss} after taking expectations, gives \eqref{sgraddss}.

Proof of Part (iv). Cancellation of the second and fourth terms in \eqref{gradCss} after taking expectations, gives \eqref{sgradCss}.
\end{proof}

\subsection{Gradient of the log-posterior for the collared-dove data model\label{app:GradientLogPosterior}}

Let $p(x|a,b)$ denote a probability density with argument $x$ and
parameters $a,b$. In what follows, these density functions are (abbreviations within parenthesis) the Inverse-Gamma (IG), the normal ($\mathcal{N}$) and the Poisson (Poisson). The log-posterior of (\ref{eq:PosteriorDistribution}),
with $\phi_{o}=\log\sigma_{o}^{2}$ for symbols $o=\epsilon,\eta,\psi,\alpha$,
is
\begin{align}
\log p(\theta|y) & =\mathrm{const}+\phi_{\epsilon}+\phi_{\eta}+\phi_{\psi}+\phi_{\alpha}\nonumber \\
 & +\log\mathrm{IG}(\exp(\phi_{\epsilon})|q_{\epsilon},r_{\epsilon})+\log\mathrm{IG}(\exp(\phi_{\eta})|q_{\eta},r_{\eta})+\log\mathrm{IG}(\exp(\phi_{\psi})|q_{\psi},r_{\psi})\nonumber \\
 & +\log\mathrm{IG}(\exp(\phi_{\alpha})|q_{\alpha},r_{\alpha})+\log \mathcal{N}(\alpha|0,\exp(\phi_{\alpha})R_{\alpha})+\log \mathcal{N}(\psi|\Phi\alpha,\exp(\phi_{\psi})I_{p})\nonumber \\
 & +\log \mathcal{N}(u_{0}|0,10I_{p})+\sum_{t=1}^{T}\log \mathcal{N}(u_{t}|G_{t-1}\psi+u_{t-1},\exp(\phi_{\eta})I_{p})\nonumber \\
 & +\sum_{t=1}^{T}\log \mathcal{N}(v_{t}|u_{t},\exp(\phi_{\epsilon})I_{p})+\sum_{t=1}^{T}\sum_{i=1}^{p}\log\mathrm{Poisson}(y_{it}|N_{it}\exp(v_{it})).\label{eq:log_posterior}
\end{align}
With $q=\mathrm{shape}$ and $r=\mathrm{scale}$,
\begin{align*}
\log\mathrm{IG}(x|q,r) & =\mathrm{const}-(q+1)\log(x)-r/x, \\
\frac{d}{dx}\log\mathrm{IG}(x|q,r) & =-(q+1)/x+r/x^{2};
\end{align*}
hence
\begin{align*}
\frac{d}{d\phi}\log\mathrm{IG}(\exp(\phi)|q,r) & =\left(-(q+1)/\exp(\phi)+r/\exp(2\phi)\right)\exp(\phi)\\
 & =-(q+1)+r/\exp(-\phi).
\end{align*}
Let $x,a$ be $p$ dimensional vectors, $b$ a scalar and $I_p$ a $p \times p $ identity matrix; then,
\[
\log \mathcal{N}(x|a,bI_{p})=\mathrm{const}-\frac{p}{2}\log(b)-\frac{1}{2b}(x-a)^{\top}(x-a);
\]
\begin{align*}
\frac{d}{dx}\log \mathcal{N}(x|a,bI) & =-\frac{1}{b}(x-a);\\
\frac{d}{da}\log \mathcal{N}(x|a,bI) & =\frac{1}{b}(x-a);\\
\frac{d}{db}\log \mathcal{N}(x|a,bI) & =-\frac{p}{2b}+\frac{1}{2b^2}(x-a)^{\top}(x-a).
\end{align*}
For
\begin{align*}
\log\mathrm{Poisson}(k|Na) & =\mathrm{const}+k\log(Na)-Na ;
\end{align*}
\begin{align*}
\frac{d}{da}\log\mathrm{Poisson}(k|Na) & =k/a-N.
\end{align*}
It is straightforward to compute the gradient
of \eqref{eq:log_posterior} using these derivatives and the chain rule.

\subsection{Log-posterior for the Wishart process model\label{app:LogPosteriorPhilipovGlickman}}
To compute the Jacobian term of the transformations in Section \ref{subsec:ModelPhilipovGlickman}, note that from standard results about the derivative of a covariance matrix with respect to its Cholesky factor
\begin{align*}
  \frac{d \mathrm{vech}(\Sigma_t)}{d \mathrm{vech}(C_t)} & =L_k(I_{k^2}+K_{k,k})(C_t\otimes I_k)L_k^\top,
\end{align*}
where $L_k$ and $K_{k,k}$ denote the elimination matrix and the commutation matrix defined in Section \ref{supp:notation};
similarly,
\begin{align*}
  \frac{d \mathrm{vech}(A)}{d \mathrm{vech}(H)} & =L_k(I_{k^2}+K_{k,k})(H\otimes I_k)L_k^\top.
\end{align*}
We also have
$$\frac{d (\nu-k)}{d\nu'}=\nu-k,\;\;\;\;\frac{d d}{d d'}=d(1-d),\;\;\;\; \frac{d C_{t,ii}}{d C'_{t,ii}}=C_{t,ii}\;\;\text{ and }\;\;\frac{d H_{ii}}{d H'_{ii}}=H_{ii};$$
hence,
\begin{align*}
  p(\theta|y) \propto & |L_k (I_{k^2}+K_{k,k})(H\otimes I_k)L_k^\top| \times \left\{\prod_{t=1}^T |L_k(I_{k^2}+K_{k,k})(C_t\otimes I_k)L_k^\top| \right\} \times (\nu-k) \\
 & \times d(1-d)\times \left\{\prod_i H_{ii}\right\}\left\{\prod_{t=1}^T \prod_{i=1}^k C_{t,ii}\right\} \times p(A,d,\nu-k)\left\{\prod_{t=1}^T p(\Sigma_t|\nu,S_{t-1},d) p(y_t|\Sigma_t)\right\}.
\end{align*}

\subsection{Gradient of the log-posterior for the Wishart process model\label{app:GradientLogPosteriorPhilipovGlickman}}
Let $h(\theta) \coloneqq p(\theta|y)$, with $p(\theta|y)$ in (\ref{eq:PosteriorDistribution_PhilipovGlickman}). First, for $t=1,\dots, T-1,$
\begin{align}
\nabla_{\mathrm{vech}(C'_t)} \log h(\theta)  = &
\nabla_{\mathrm{vech}(C'_t)} \log |L_k(I_{k^2}+K_{k,k})(C_t\otimes I_k)L_k^\top|+\sum_{i=1}^k \nabla_{\mathrm{vech}(C'_t)} \log C_{t,ii} \nonumber \\
 & +\nabla_{\mathrm{vech}(C'_t)} \log p(\Sigma_t|\nu,S_{t-1},d)+\nabla_{\mathrm{vech}(C'_t)} \log p(\Sigma_{t+1}|\nu,S_t,d) \nonumber \\
& +\nabla_{\mathrm{vech}(C'_t)} \log p(y_t|\Sigma_t) \nonumber \\
= & T_{t1}+T_{t2}+T_{t3}+T_{t4}+T_{t5}; \label{Ctterms}
\end{align}
the expression is the same for $t=T$, with the fourth term $T_{t4}$ omitted.

We now define $T_{t1}, \dots , T_{t1}$ in (\ref{Ctterms}).
\begin{align*}
  T_{t1} = & \nabla_{\mathrm{vech}(C'_t)} \log |L_k(I_{k^2}+K_{k,k})(C_t\otimes I_k)L_k^\top| \\
 = & \mathrm{vech}(D(C_t))\odot L_k (I_k\otimes \{(I_k\otimes \mathrm{vec}(I_k)^\top)(K_{k,k}\otimes I_k)\}) \\
 & \times (L_k^\top\otimes (I_{k^2}+K_{k,k})L_k^\top)
\mathrm{vec}(\{L_k (I_{k^2}+K_{k,k})(C_t\otimes I_k)L_k^\top\}^{-\top});
\end{align*}
for a square matrix $A$, $D(A)$ is a square matrix of the same dimension as $A$, and
having all entries 1, except for the diagonal entries which are equal to the corresponding diagonal entries of $A$.
The derivation of the above expression follows from the chain rule; the standard results
\begin{align*}
\nabla_{\mathrm{vec}(A)} \log |A|=\mathrm{vec}(A^{-\top}),\;\;\;\; \frac{d AXB}{d X}= & B^\top\otimes A,\;\;\;\;\frac{d \mathrm{vech}(C_t)}{d\mathrm{vec}(C_t)}=L_k^\top;
\end{align*}
the observation that
$d\mathrm{vech}(C_t)/d\mathrm{vech}(C'_t)$ is the diagonal matrix with diagonal entries $\mathrm{vech}(D(C_t));$ and
that by Theorem~11 of \citet{magnus+n85},
\begin{align*}
  \frac{d C_t\otimes I_k}{d C_t}=(I_k\otimes \left\{(K_{k,k}\otimes I_k)(I_k\otimes \mathrm{vec}I_k))\right\}.
\end{align*}
These results, together with the identities
\begin{align*}
  \frac{d A^{-1}}{d A} & = - (A^{-\top}\otimes A^{-1}) \;\;\;\;\text{ and }\;\;\;\;\;\frac{d \mathrm{tr}(AB)}{d B}=\mathrm{vec}(A^\top)^\top,
\end{align*}
are used repeatedly in the derivations below.
Next,
\begin{align*}
  T_{t2} = & \sum_{i=1}^k \nabla_{\mathrm{vech}(C'_t)} \log C_{t,ii} = \mathrm{vech}(I_k);
\end{align*}
\small{\begin{align*}
  T_{t3} = & \nabla_{\mathrm{vech}(C'_t)} \log p(\Sigma_t|\nu,S_{t-1},d) \\
= & \nabla_{\mathrm{vech}(C'_t)} \left\{-\frac{\nu+k+1}{2}\log |\Sigma_t|-\frac{1}{2}\mathrm{tr}(S_{t-1}^{-1}\Sigma_t^{-1})\right\} \\
= & \mathrm{vech}(D(C_t))\odot \left\{L_k(I_{k^2}+K_{k,k})(C_t\otimes I_k)L_k^\top\right\}^\top D_k^\top \left\{-\frac{\nu+k+1}{2}\mathrm{vec}(\Sigma_t^{-1})
+\frac{1}{2}(\Sigma_t^{-1}\otimes \Sigma_t^{-1})\mathrm{vec}(S_{t-1}^{-1})\right\} \\
= & \mathrm{vech}(D(C_t))\odot \left\{L_k(C_t^\top\otimes I_k)(I_{k^2}+K_{k,k})L_k^\top\right\} D_k^\top \left\{-\frac{\nu+k+1}{2}\mathrm{vec}(\Sigma_t^{-1})+\frac{1}{2}(\Sigma_t^{-1}\otimes \Sigma_t^{-1})\mathrm{vec}(S_{t-1}^{-1})\right\},
\end{align*}}
where $D_k$ is the duplication matrix in Section \ref{supp:notation}; and
\small{\begin{align*}
 T_{t4} = & \nabla_{\mathrm{vech}(C'_t)} \log p(\Sigma_{t+1}|\nu,S_t,d) \\
 = & \nabla_{\mathrm{vech}(C'_t)} \left\{ -\frac{\nu}{2}\log |S_t|-\frac{1}{2}\mathrm{tr}(S_t^{-1}\Sigma_{t+1}^{-1})\right\} \\
 = & \mathrm{vech}(D(C_t))\odot \{L_k(C_t^\top\otimes I_k)(I_{k^2}+K_{k,k})L_k^\top\}D_k^\top\left\{\frac{d S_t}{d\Sigma_t}\right\}^\top\left\{-\frac{\nu}{2}\mathrm{vec}(S_t^{-1})+\frac{1}{2}(S_t^{-\top}\otimes S_t^{-1})\mathrm{vec}(\Sigma_{t+1}^{-1})\right\}.
\end{align*}}
Above,
\begin{align*}
  \frac{d S_t}{d\Sigma_t} = & \frac{1}{\nu}(H\otimes H)\left\{\frac{d \Sigma_t^{-d}}{d\Sigma_t}\right\};
\end{align*}
and define,
\begin{align}
  \Sigma_t^{-d} := & P_t\Lambda_t^{-d} P_t^\top, \label{powerdef}
\end{align}
where $\Sigma_t=P_t\Lambda_t P_t^\top$ is the eigen-decomposition of $\Sigma_t$ with
$P_t$ the orthonormal matrix of the eigenvectors and $\Lambda_t$ is
the diagonal matrix of the eigenvalues; we denote the $j$th column of $P_t$
(the $j$th eigenvector) as $p_{tj}$, and the $j$th diagonal element of $\Lambda_t$ (the $j$th eigenvalue) as $\lambda_{tj}$, where $\lambda_{t1}>\dots >\lambda_{tk}>0$
and $\Lambda_t^{-d}$ is the diagonal matrix with $j$th diagonal entry $\lambda_{tj}^{-d}$.

Writing,
$\Sigma_t^{-d}=\sum_{i=1}^k (\lambda_{ti}^{-d}I_k) p_{ti}p_{ti}^\top$, and  using the product rule,
\begin{align*}
  \frac{d \Sigma_t^{-d}}{d \Sigma_t} = & \sum_{i=1}^k \left\{(p_{ti}p_{ti}^\top\otimes I_k)\frac{d \lambda_{ti}^{-d}I_k}{d\Sigma_t}+(I_k\otimes \lambda_{ti}^{-d}I_k)
 \frac{ d p_{ti}p_{ti}^\top}{d\Sigma_t}\right\},
\end{align*}
 where
\begin{align*}
 \frac{d\lambda_{ti}^{-d} I_k}{d \Sigma_t} = & -d \lambda_{ti}^{-d-1} \mathrm{vec}(I_k)\frac{d \lambda_{ti}}{d \Sigma_t}
 =  -d \lambda_{ti}^{-d-1}\mathrm{vec}(I_k)(p_{ti}^\top\otimes p_{ti}^\top);
\end{align*}
the last line follows from Theorem 1 of \cite{magnus85}.
Moreover,
\begin{align*}
\frac{ d p_{ti}p_{ti}^\top}{d\Sigma_t} = & \frac{d p_{ti}p_{ti}^\top}{d p_{ti}} \frac{d p_{ti}}{d\Sigma_t}
 =  \left\{p_{ti}\otimes I_k+I_k\otimes p_{ti}\right\} \times \left\{p_{ti}^\top\otimes (\lambda_{ti}I_k-\Sigma_t)^{-}\right\},
\end{align*}
where $A^-$ denotes the Moore-Penrose inverse of $A$ and using Theorem 1 in \cite{magnus85}.
Finally,
\begin{align*}
  T_{t5} = & \nabla_{\mathrm{vech}(C'_t)} \log p(y_t|\Sigma_t) \\
 = & \nabla_{\mathrm{vech}(C'_t)} \left\{-\frac{1}{2}\log |\Sigma_t|-\frac{1}{2}y_t^\top \Sigma_t^{-1} y_t \right\}  \\
% = &  \left\{ \frac{d \mathrm{vech}(C_t)}{d \mathrm{vech}{C_t'}}\right\}^\top\left\{ \frac{d \mathrm{vech}(\Sigma_t)}{d \mathrm{vech}(C_t)} \right\}^\top
%\left\{\frac{d \mathrm{vec}(\Sigma_t)}{d \mathrm{vech}(\Sigma_t)}\right\}^\top \nabla_{\mathrm{vec}(\Sigma_t)} \left\{-\frac{1}{2}\log \Sigma_t-\frac{1}{2}y_t^\top\Sigma_t^{-1}y_t\right\}  \\
 = & \mathrm{vech}(D(C_t))\odot \left\{L_k(I_{k^2}+K_{k,k})(C_t\otimes I_k)L_k^\top\right\}^\top D_k^\top \left\{-\frac{1}{2}\mathrm{vec}(\Sigma_t^{-1})+\frac{1}{2}(\Sigma_t^{-1}\otimes \Sigma_t^{-1})(y_t\otimes y_t)\right\} \\
 = & \mathrm{vech}(D(C_t))\odot \left\{L_k(C_t^\top\otimes I_k)(I_{k^2}+K_{k,k})L_k^\top\right\} D_k^\top \left\{-\frac{1}{2}\mathrm{vec}(\Sigma_t^{-1})+\frac{1}{2}(\Sigma_t^{-1}y_t \otimes \Sigma_t^{-1}y_t)\right\}.
\end{align*}
Next, consider
\begin{align*}
  \nabla_{\mathrm{vech}(H')} \log h(\theta) = & \nabla_{\mathrm{vech}(H')} \log p(A)+\nabla_{\mathrm{vech}(H')} \sum_{t=1}^T \log p(\Sigma_t|\nu,S_{t-1},d) \\
   & +\nabla_{\mathrm{vech}(H')} \sum_{i=1}^k \log H_{ii}+\nabla_{\mathrm{vech}(H')} \log |L_k(I_{k^2}+K_{k,k})(H\otimes I_k)L_k^\top| \\
 = & T_{H1}+T_{H2}+T_{H3}+T_{H4};
\end{align*}
We give expressions for $T_{H1}, \dots, T_{H4}$ below.
\small{\begin{align*}
  T_{H1}= & \nabla_{\mathrm{vech}(H')} \log p(A) \\
 = & \mathrm{vech}(D(H))\odot L_k(H^\top \otimes I_k)(I_{k^2}+K_{k,k})L_k^\top \times D_k^\top\left(-\frac{\gamma_0+k+1}{2}\mathrm{vec}(A^{-1})+\frac{1}{2}(A^{-1}\otimes A^{-1})\mathrm{vec}(Q_0^{-1})\right),
\end{align*}}
with the derivation similar to that of $T_{t3}$.
\begin{align*}
 T_{H2} = & \sum_{t=1}^T \nabla_{\mathrm{vech}(H')}\log p(\Sigma_t|\nu,S_{t-1},d) \\
 = & \sum_{t=1}^T  \left\{-\frac{\nu}{2}\nabla_{\mathrm{vech}(H')} \log |S_{t-1}|-\frac{1}{2}\nabla_{\mathrm{vech}(H')} \mathrm{tr}(S_{t-1}^{-1}\Sigma_t^{-1})\right\},
\end{align*}
with
\begin{align*}
  \nabla_{\mathrm{vech}(H')} \log |S_{t-1}| = & \nabla_{\mathrm{vech}(H')} \log |\frac{1}{\nu} H\Sigma_{t-1}^{-d} H^\top| \\
 = & 2 \nabla_{\mathrm{vech}(H')} \log |H| \\
 = & 2 \mathrm{vech}(D(H))\odot L_k \mathrm{vec}(H^{-\top});
\end{align*}
and
\begin{align*}
  \nabla_{\mathrm{vech}(H')} \mathrm{tr}(S_{t-1}^{-1}\Sigma_t^{-1}) = &
\left\{ \frac{d \mathrm{tr}(S_{t-1}^{-1}\Sigma_t^{-1})}{d S_{t-1}} \frac{d S_{t-1}}{d H} \frac{d H}{d \mathrm{vech}(H)} \frac{d \mathrm{vech}(H)}{d \mathrm{vech}(H')}\right\}^\top \\
 = & -\mathrm{vech}(D(H))\odot L_k \left\{\frac{d S_{t-1}}{d H}\right\}^\top (S_{t-1}^{-\top}\otimes S_{t-1}^{-1}) \mathrm{vec}(\Sigma_t^{-1});
\end{align*}
where
\begin{align*}
 \frac{d S_{t-1}}{d H} = & \frac{1}{\nu}\frac{d H\Sigma_{t-1}^{-d} H^\top}{d H} = \frac{1}{\nu} \frac{ d (H\Sigma_{t-1}^{-d/2}\Sigma_{t-1}^{-d/2}H^\top)}{d H\Sigma_{t-1}^{-d/2}}
 \frac{d H\Sigma_{t-1}^{-d/2}}{d H} \\
 = & \frac{1}{\nu} (I_{k^2}+K_{k,k})(H\Sigma_{t-1}^{-d/2}\otimes I_k)(\Sigma_{t-1}^{-d/2}\otimes I_k)=\frac{1}{\nu}(I_{k^2}+K_{k,k})(H\Sigma_{t-1}^{-d}\otimes I_k).
\end{align*}
\begin{align*}
  T_{H3} = & \sum_{i=1}^k \nabla_{\mathrm{vech}(H')} \log H_{ii} = \mathrm{vech}(I_k).
\end{align*}
\begin{align*}
  T_{H4} = & \nabla_{\mathrm{vech}(H')} \log |L_k(I_{k^2}+K_{k,k})(H\otimes I_k)L_k^\top| \\
 = & \mathrm{vech}(D(H))\odot L_k(I_k\otimes \{(I_k\otimes \mathrm{vec}(I_k)^\top)(K_{k,k}\otimes I_k)\}) \\
 & \times (L_k^\top\otimes (I_{k^2}+K_{k,k})L_k^\top)
\mathrm{vec}(\{L_k (I_{k^2}+K_{k,k})(H\otimes I_k)L_k^\top\}^{-\top}).
\end{align*}
by a similar derivation to that $T_{t1}$.
Next, consider the gradient for $d'$.
\begin{align*}
  \frac{d \log h(\theta)}{dd'} = & \frac{d}{dd} \log d(1-d) \frac{dd}{dd'} +\frac{d \log p(d)}{dd'} + \frac{d}{d d'}\sum_{t=1}^T \log p(\Sigma_t|\nu,S_{t-1},d)
 = & T_{d 1}+T_{d 2}+T_{d 3}.
\end{align*}
Here $T_{d 1} =1-2d$, $T_{d 2}=0$ and
\begin{align*}
  T_{d 3} = & \sum_{t=1}^T \frac{d}{d S_{t-1}} \left\{ -\frac{\nu}{2}\log |S_{t-1}|-\frac{1}{2}\mathrm{tr}(S_{t-1}^{-1}\Sigma_t^{-1})\right\}\times \frac{d S_{t-1}}{d d} \frac{d d}{d d'} \\
 = & \sum_{t=1}^T \left\{ -\frac{\nu}{2}\mathrm{vec}(S_{t-1}^{-1})+\frac{1}{2}(S_{t-1}^{-\top}\otimes S_{t-1}^{-1})\mathrm{vec}(\Sigma_t^{-1})\right\}^\top \\
   & \times \left\{-\frac{1}{\nu}(H\otimes H)\sum_{i=1}^k \log \lambda_{t-1,i}(\lambda_{t-1,i}^{-d})\mathrm{vec}(p_{t-1,i}p_{t-1,i}^\top)\right\} \times d(1-d).
\end{align*}

Finally, consider the gradient for $\nu'$,
\begin{align*}
 \frac{d}{d\nu'} \log h(\theta) = & \frac{d \log (\nu-k)}{d\nu'}+\frac{d}{d\nu'} \log p(\nu-k)+\sum_{t=1}^T \frac{d}{d\nu'} \log p(\Sigma_t|\nu,S_{t-1},d) \\
 = & T_{\nu 1}+T_{\nu 2}+T_{\nu 3},
\end{align*}
where $T_{\nu 1}=1$,
\begin{align*}
  T_{\nu 2}= & (\nu-k) \frac{d}{d\nu}\left\{ (\alpha_0-1)\log (\nu-k)-\beta_0 (\nu-k)\right\} \\
 = & (\alpha_0-1)-\beta_0 (\nu-k)
\end{align*}
and
\small{\begin{align*}
  T_{\nu 3} = & (\nu-k) \sum_{t=1}^T \frac{d}{d\nu}\log p(\Sigma_t|\nu,S_{t-1},d) \\
 = & (\nu-k) \sum_{t=1}^T \frac{d}{d\nu}\left\{ -\frac{\nu k}{2}\log 2 -\sum_{i=1}^k \log \Gamma\left(\frac{\nu+1-i}{2}\right)-\frac{\nu}{2}\log |S_{t-1}|-\frac{\nu+k+1}{2}\log |\Sigma_t|-\frac{1}{2}\mathrm{tr}(S_{t-1}^{-1}\Sigma_t^{-1})\right\} \\
 = & (\nu-k) \left\{-\frac{Tk}{2}\log 2 -\frac{T}{2}\sum_{i=1}^k \psi\left(\frac{\nu+1-i}{2}\right) \right. \\
  & \left.-\sum_{t=1}^T \left\{ -\frac{k}{2}+\frac{1}{2}\log |S_{t-1}\Sigma_t|+\frac{1}{2}\mathrm{tr}\left((H^\top)^{-1}\Sigma_{t-1}^{-d}H^{-1}\Sigma_t^{-1} \right)\right\}\right\}.
\end{align*}}

\section{Low-Rank State and Auxiliary variable (LR-SA): Including the auxiliary variable in
the low rank approximation}
\label{sec:LRSA_derivations}

The use of the LR-SA parametrization for the model of Section 5 is now described.
Following Section \ref{sec:Gaussian_Variational_high_dim_state},   the $p$ dimensional vectors $u_{t}$
and $v_{t}$ are modeled through the lower dimensional factor \begin{eqnarray*}
X_{t}^{(1)} & = & B_{1}z_{t}^{(1)}+\epsilon_{t}^{(1)},\quad\epsilon_{t}^{(1)}\sim \mathcal{N}\left(0,\left(D_{t}^{(1)}\right)^{2}\right)\\
X_{t}^{(2)} & = & B_{2}z_{t}^{(2)}+\epsilon_{t}^{(2)},\quad\epsilon_{t}^{(2)}\sim \mathcal{N}\left(0,\left(D_{t}^{(2)}\right)^{2}\right),
\end{eqnarray*}
where  $z_{t}^{(1)}$ and $z_{t}^{(2)}$ are $q$ dimensional, $q \ll p$;
$B_{1}$ and $ B_{2}$ are $p\times q$ matrices; and both $D^{(1)} _{t}$ and $D^{(1)} _{t}$ are $p \times p$ matrices.
Define,
\begin{align*}
z^{(1)}&:=\vec \left (z_{0}^{(1)},\dots,z_{T}^{(1)}\right ) ,
z^{(2)}:=\vec\left (z_{1}^{(2)},\dots,z_{T}^{(2)}\right ),\\
\zeta&:=\vec(\psi,\alpha,\phi_{\epsilon},\phi_{\eta},\phi_{\psi},\phi_{\alpha})\quad \text{and} \quad %
\rho:=\vec(z^{(1)},z^{(2)},\zeta).
\end{align*}
Let $\mu_{1}=E(z^{(1)}),\mu_{2}=E(z^{(2)})$,
$\mu_{3}=E(\zeta)$, and $\mu= \vec ( \mu_{1},\mu_{1}, \mu_{1}   ) $
Let $C_{1}^{(1)}$ and $C_{1}^{(2)}$ denote the variational parameters that model the
precision matrices of $(z^{(1)}$ and $z^{(2)})$;  $C_{1}^{(1)}$ is identical
to the matrix $C_{1}$ described in Section \ref{sec:Gaussian_Variational_high_dim_state} (which gives a band-structure for
$\Omega_{z^{(1)}}$); we take $C_{1}^{(2)}$ as block-diagonal
with $T$ blocks, where each block models the precision matrix of
$z_{t}^{(2)}$ at time $t=1,\dots,T$. We form $C_{1}$ by combining
$C_{1}^{(1)}$ and $C_{1}^{(2)}$ as a block-diagonal matrix;
however, we let it be non-zero for the part corresponding to the correlation
between $u_{t}$ and $v_{t}$ at time $t$, and zero otherwise since
$u_{i}$ and $v_{j}$ are conditionally independent for $i\neq j$.
The Cholesky factor of the precision matrix for $\rho$
is then $C=\mathrm{block}(C_{1},C_{2})$, where $C_{2}$ is specified
similarly to Section \ref{sec:Gaussian_Variational_high_dim_state}, but omitting the dependencies that include $v$
since it is now in the $z$-block and the derivations in Section \ref{sec:Gaussian_Variational_high_dim_state} assume that $z$ is independent of $\zeta$. The reparameterization trick is then applied using the transformation
\begin{eqnarray*}
\theta & = & \widetilde{W}\rho+\widetilde{Z}e=\widetilde{W}\mu+\widetilde{W}C^{-1}\omega+\widetilde{Z}e,
\end{eqnarray*}
where
\[
e=\begin{bmatrix}\epsilon\\
0_{P\times1}
\end{bmatrix},\quad\epsilon\sim \mathcal{N}\left(0,I_{p(T+1)+pT}\right),\omega=\begin{bmatrix}\omega_{1}\\
\omega_{2}\\
\omega_{3}
\end{bmatrix}\sim \mathcal{N}\left(0,I_{q(T+1)+qT+P}\right);
\]
\[
\widetilde{W}=\begin{bmatrix}I_{T+1}\otimes B_{1} & 0_{p(T+1)\times qT} & 0_{p(T+1)\times P}\\
0_{pT\times q(T+1)} & I_{T}\otimes B_{2} & 0_{pT\times P}\\
0_{P\times q(T+1)} & 0_{P\times qT} & I_{P\times P}
\end{bmatrix}\text{ and }\widetilde{Z}=\begin{bmatrix}D_{t}^{(1)} & 0_{p(T+1)\times pT} & 0_{p(T+1)\times P}\\
0_{pT\times p(T+1)} & D_{t}^{(2)} & 0_{pT\times P}\\
0_{P\times p(T+1)} & 0_{P\times pT} & 0_{P\times P}
\end{bmatrix}.
\]
The gradients for $\mu,C$ and $D$ follow immediately
from the previous derivations. However, this does not
apply to the gradient for $B$ because $z^{(2)}$ has $T$ observations
and not $T+1$. Writing $\widetilde{W}=W_{1}+W_{2}+W_{3}$ with
\[
W_{1}=\begin{bmatrix}I_{T+1}\otimes B_{1} & 0_{p(T+1)\times qT} & 0_{p(T+1)\times P}\\
0_{pT\times q(T+1)} & 0_{pT\times qT} & 0_{pT\times P}\\
0_{P\times q(T+1)} & 0_{P\times qT} & 0_{P\times P}
\end{bmatrix},W_{2}=\begin{bmatrix}0_{p(T+1)\times q(T+1)} & 0_{p(T+1)\times qT} & 0_{p(T+1)\times P}\\
0_{pT\times q(T+1)} & I_{T}\otimes B_{2} & 0_{pT\times P}\\
0_{P\times q(T+1)} & 0_{P\times qT} & 0_{P\times P}
\end{bmatrix}
\]
and
\[
W_{3}=\begin{bmatrix}0_{p(T+1)\times q(T+1)} & 0_{p(T+1)\times qT} & 0_{p(T+1)\times P}\\
0_{pT\times q(T+1)} & 0_{pT\times qT} & 0_{pT\times P}\\
0_{P\times q(T+1)} & 0_{P\times qT} & I_{P\times P}
\end{bmatrix}.
\]
Then, $d \wt W_1/dB_1 = dW_1/dB_1$ because $W_2$ and $W_3$ do not depend on $B_1$, and
\begin{eqnarray*}
 \frac{dW_1}{dB_1} & = & (Q_{1}^{(1)^{\top}}\otimes P_{1}^{(1)})\left[\left\{ (I_{T+1}\otimes K_{q(T+1)})(\mathrm{vec}(I_{T+1})\otimes I_{q})\right\} \otimes I_{p}\right],
\end{eqnarray*}
using \eqref{dvecWdvecB},
with
\[
P_{1}^{(1)}=\begin{bmatrix}I_{p(T+1)}\\
0_{pT\times p(T+1)}\\
0_{P\times p(T+1)}
\end{bmatrix}\quad\text{and }Q_{1}^{(1)}=\begin{bmatrix}I_{q(T+1)} & 0_{q(T+1)\times qT} & 0_{q(T+1)\times P}\end{bmatrix}.
\]
Similarly,  $d\widetilde{W}/dB_2 = dW_2/dB_2$ because $W_{1}$ and $W_{3}$ do not depend on $B_{2}$, with
\begin{eqnarray*}
\frac{dW_2}{dB_2} & = & (Q_{1}^{(2)^{\top}}\otimes P_{1}^{(2)})\left[\left\{ (I_{T+1}\otimes K_{q(T+1)})(\mathrm{vec}(I_{T+1})\otimes I_{q})\right\} \otimes I_{p}\right],
\end{eqnarray*}
and
\[
P_{1}^{(2)}=\begin{bmatrix}0_{p(T+1)\times pT}\\
I_{pT}\\
0_{P\times pT}
\end{bmatrix}\quad\text{and }Q_{1}^{(2)}=\begin{bmatrix}0_{qT\times q(T+1)} & I_{qT} & 0_{qT\times P}\end{bmatrix}.
\]

\section{Some further discussion of  the Wishart process model}
\label{sec:Wishart_details}
%This section first studies the inefficiency of MCMC for the model in \citet{philipov+g06} closely and then
This section derives the oracle  and variational predictive densities used for evaluating the proposed
variational approach; it also augments the comparison in Section~\ref{subsec:Results_PhilipovGlickman} of the variational and
 oracle predictive densities.

\subsection{The oracle predictive density \label{subsec:OraclePredictive}}
Let $\zeta$ be the static model parameter for the model in Section \ref{subsec:ModelPhilipovGlickman},
and $\zeta^{\mathrm{true}}$ its true value in the simulation;
$\zeta$ only appears in the state equation.
The one-step ahead oracle predictive density $p(y_{T + 1}|y_{1:T}, \zeta^{\mathrm{true}})$ is obtained empirically using
simulation by averaging over the states. Using conditional independence,
\begin{align}\label{eq:OraclePredictive}
p(y_{T + 1},X_{T+1},X_T|y_{1:T}, \zeta^{\mathrm{true}}) & = p(y_{T + 1}|X_{T+1}, \zeta^{\mathrm{true}})
p(X_{T+1}|X_T,y_{1:T}, \zeta^{\mathrm{true}})p(X_T|y_{1:T}, \zeta^{\mathrm{true}});
\end{align}
For $j=1, \dots, M$, we  use \eqref{eq:OraclePredictive} and the bootstrap particle filter \citep{gordon1993novel} 
to obtain $\{X_T^{(j)}\} $ from $p(X_T|y_{1:T}, \zeta^{\mathrm{true}})$;  we then use \eqref{eq:OraclePredictive} to
 generate $X_{T+1}^{(j)}$ from $p(X_{T+1}|X_T^{(j)},y_{1:T}, \zeta^{\mathrm{true}})$; and then $Y_{T+1}^{(j)}$ from $p(y_{T + 1}|X_{T+1}^{(j)}, \zeta^{\mathrm{true}})$. 
 The predictive density is then estimated by $\{Y_{T+1}^{(j)}\}_{j=1}^M $, each having weight $1/M$. 

\subsection{The variational predictive density \label{subsec:VariationalPredictive}}
The one-step ahead variational predictive density $p(y_{T + 1}|y_{1:T})$
averages over both the states and the static model parameters using the variational posterior.
Similarly to \eqref{eq:OraclePredictive}, we write
\begin{align}\label{eq:VariationalPredictive}
p(y_{T + 1},X_{T+1},X_T, \zeta|y_{1:T}) & =
p(y_{T + 1}|X_{T+1}, \zeta)
p(X_{T+1}|X_T,y_{1:T}, \zeta)p(X_T,\zeta|y_{1:T});
\end{align}
 For $j=1, \dots, M$, we generate  $\{X_T^{(j)},\zeta{(j)}\}_{j=1}^M $ from $q(X_T,\zeta)$ --- which is the variational approximation of $p(X_T,\zeta|y_{1:T})$;
 we
 then generate $ X_{T+1}^{(j)}$ from $p(X_{T+1}|X_T^{(j)},y_{1:T}, \zeta^{(j)})$;
 and then 
$Y_{T+1}^{(j)}$ from \eqref{subsec:VariationalPredictive}.
 The predictive density is then estimated by $\{Y_{T+1}^{(j)}\}_{j=1}^M .$

\subsection{Further results \label{subsec:VariationalPredictiveEvaluationsAppendix}}
Section \ref{subsec:Results_PhilipovGlickman} assesses the out-of-sample predictive properties by comparing the variational predictive density to the oracle predictive density. For $T = 100$,
Figure~\ref{fig:OracleVsVariational} 
in the paper shows the accuracy for all five marginal one-step ahead predictive densities as well as
%all ten bivariate predictive densities. Figure~\ref{fig:OracleVsVariational_boxplots} shows the one-step ahead marginal
predictive densities, both the variational and the oracle, for four time points $T=100,101,102,103$ and all five variables. Figure \ref{fig:OracleVsVariational_bivariates_random_selection} complements these figures by showing a subset of 15 bivariate one-step ahead predictive densities for $T=101,102,103$; we have verified similar accuracies for all the other bivariate posteriors.

\begin{figure}[h!]
\centering
\includegraphics[width=0.9\columnwidth]{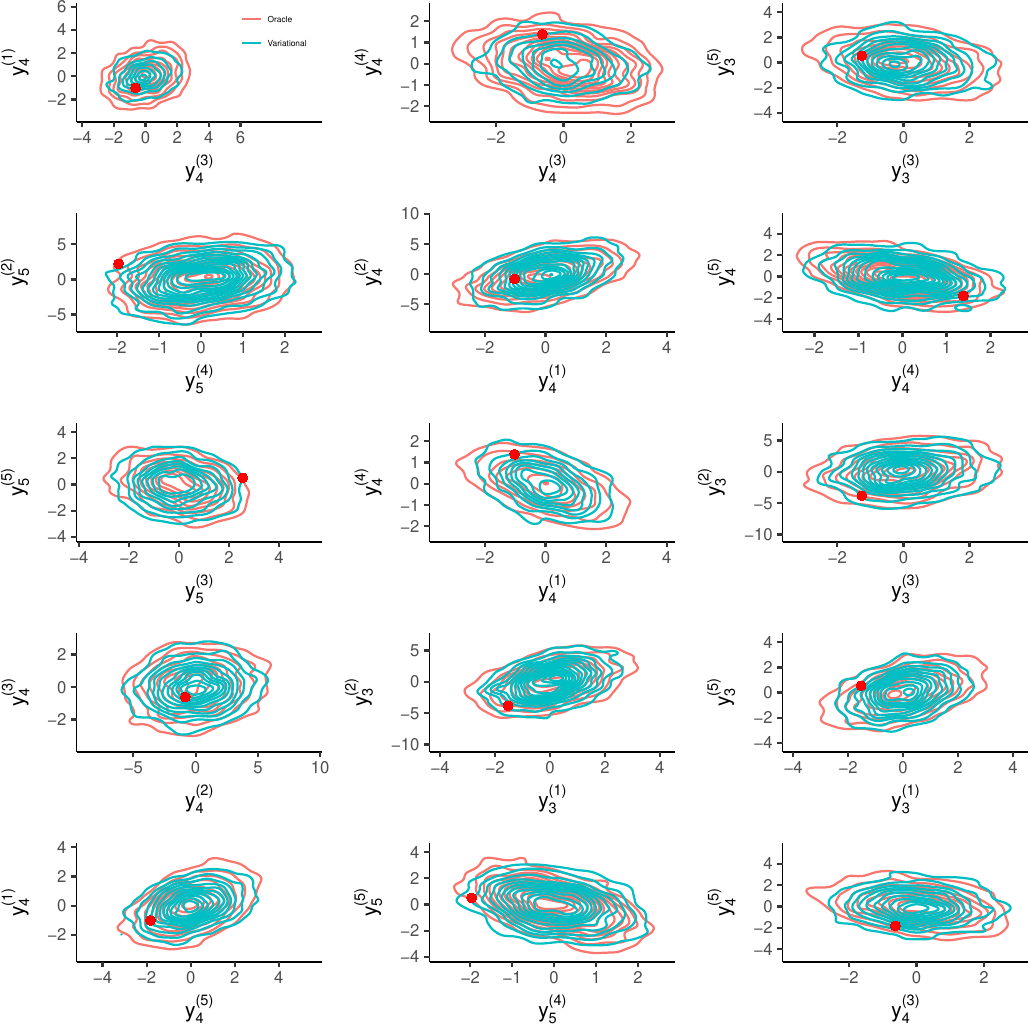}
\caption{The variational one-step ahead predictive density against the oracle one-step ahead predictive density for simulated data as in Section \ref{subsec:Results_PhilipovGlickman}. The one-step ahead bivariate predictive densities are shown together with the test observation (red dot). The labels of the axis show which $T + 1$ and pair of variables is considered.}
\label{fig:OracleVsVariational_bivariates_random_selection}
\end{figure}

\section{Details on the parsimony of the VB parametrization in the applications. \label{app:Sparsity}}
%This section gives further details of the number of Gaussian variational parameters in the different parts of the variational %structure assumed in Section~\ref{sec:Gaussian_Variational_high_dim_state}.
Tables \ref{tab:VB_parameterization} and \ref{tab:VB_parameterization_PhilipovGlickman} provides further details of the number of Gaussian variational parameters in the different parts of the variational structure for the spatio-temporal model and the Wishart process example, respectively.
\begin{table}[htbp]
\centering \caption{\emph{Parsimony of different VB parametrizations in the spatio-temporal model}. The table shows the number of variational parameters in the different VB parametrizations obtained by combining
either low-rank state / low-rank state and auxiliary (LR-S / LR-SA)
with either of low-dimensional state mean / high-dimensional state
mean (LD-SM / HD-SM).
The variational parameters for the
different VB parametrizations (with $T=18$) are divided into $\mu,B,D,C_{1}$
and $C_{2}$ defined in Section \ref{sec:Gaussian_Variational_high_dim_state}. The saturated Gaussian variational approximation has $8,923,199$ parameters when $\theta$ is 4223 dimensional.}

\begin{tabular}{llcccccccccccc}
\toprule
\textbf{\footnotesize{}Parametrization} &    & {\footnotesize{}$\mu$} &  & {\footnotesize{}$B$} &  & {\footnotesize{}$D$} &  & {\footnotesize{}$C_{1}$} &  & {\footnotesize{}$C_{2}$} &  & {\footnotesize{}$\mathrm{Total}$} & \tabularnewline
\cmidrule{1-1} \cmidrule{3-3} \cmidrule{5-5} \cmidrule{7-7} \cmidrule{9-9} \cmidrule{11-11} \cmidrule{13-13}
{\footnotesize{}LR-S + LD-SM} &    & {\footnotesize{}$2,190$} &  & {\footnotesize{}$438$} &  & {\footnotesize{}$2,109$} &  & {\footnotesize{}$370$} &  & {\footnotesize{}$4,447$} &  & {\footnotesize{}$9,554$} & \tabularnewline
{\footnotesize{}LR-S + HD-SM} & & {\footnotesize{}$4,223$} &  & {\footnotesize{}$438$} &  & {\footnotesize{}$2,109$} &  & {\footnotesize{}$370$} &  & {\footnotesize{}$4,447$} &  & {\footnotesize{}$11,587$} & \tabularnewline
{\footnotesize{}LR-SA + LD-SM} &   & {\footnotesize{}$264$} &  & {\footnotesize{}$876$} &  & {\footnotesize{}$4,107$} &  & {\footnotesize{}$730$} &  & {\footnotesize{}$451$} &  & {\footnotesize{}$6,428$} & \tabularnewline
{\footnotesize{}LR-SA + HD-SM} &    & {\footnotesize{}$4,223$} &  & {\footnotesize{}$876$} &  & {\footnotesize{}$4,107$} &  & {\footnotesize{}$730$} &  & {\footnotesize{}$451$} &  & {\footnotesize{}$10,387$} & \tabularnewline
\bottomrule
\end{tabular}\label{tab:VB_parameterization}
\end{table}

\begin{table}[htbp]
\centering \caption{\emph{Parsimony of different VB parametrizations in the Wishart process example}. The table shows the number of variational parameters in the different VB parametrizations obtained by the low-dimensional state mean / high-dimensional state
mean (LD-SM / HD-SM).
The variational parameters for the different VB parametrizations (with $T=100$) are divided into $\mu,B,D,C_{1}$
and $C_{2}$ defined in Section \ref{sec:Gaussian_Variational_high_dim_state} for two different examples. The first example is
 model $\mathrm{M}_1$ which is benchmarked against MCMC which has $k=5$ ($p = 15$). The second example,
 $\mathrm{M}_2$,  has $k=12$ ($p = 78$), with only the LD-SM parametrization considered. The saturated Gaussian variational approximation for  model $\mathrm{M}_1$ has $1,152,920$ variational parameters with $\theta$ ${1517}$ dimensional.
 For  $model \mathrm{M}_2$,  the corresponding number of variational parameters is $31,059,020$ with
 with $\theta$  ${7880}$ dimensional.}

\begin{tabular}{llcccccccccccc}
\toprule
\textbf{\footnotesize{}Parametrization} &    & {\footnotesize{}$\mu$} &  & {\footnotesize{}$B$} &  & {\footnotesize{}$D$} &  & {\footnotesize{}$C_{1}$} &  & {\footnotesize{}$C_{2}$} &  & {\footnotesize{}$\mathrm{Total}$} & \tabularnewline
\cmidrule{1-1} \cmidrule{3-3} \cmidrule{5-5} \cmidrule{7-7} \cmidrule{9-9} \cmidrule{11-11} \cmidrule{13-13}
{\footnotesize{}$\mathrm{M}_1$ - LD-SM} &   & {\footnotesize{}$417$} &  & {\footnotesize{}$54$} &  & {\footnotesize{}$1,500$} &  & {\footnotesize{}$1,990$} &  & {\footnotesize{}$48$} &  & {\footnotesize{}$4,009$} & \tabularnewline
{\footnotesize{}$\mathrm{M}_1$ - HD-SM} &   & {\footnotesize{}$1,517$} &  & {\footnotesize{}$54$} &  & {\footnotesize{}$1,500$} &  & {\footnotesize{}$1,990$} &  & {\footnotesize{}$48$} &  & {\footnotesize{}$5,109$} & \tabularnewline
{\footnotesize{}$\mathrm{M}_2$ - LD-SM} &   & {\footnotesize{}$480$} &  & {\footnotesize{}$306$} &  & {\footnotesize{}$7,800$} &  & {\footnotesize{}$1,990$} &  & {\footnotesize{}$237$} &  & {\footnotesize{}$10,813$} & \tabularnewline

%{\footnotesize{}$\mathrm{M}_2$ - HD-SM} &  & {\footnotesize{}$\cdot$}  &  & {\footnotesize{}$\cdot$} &  & {\footnotesize{}$\cdot$} &  & {\footnotesize{}$\cdot$} &  & {\footnotesize{}$\cdot$} &  & {\footnotesize{}$\cdot$} &  & {\footnotesize{}$\cdot$} & \tabularnewline
\bottomrule
\end{tabular}\label{tab:VB_parameterization_PhilipovGlickman}
\end{table}

\section{The problem with the  \citet{philipov+g06} MCMC implementation of their stochastic volatility model \label{subsec:ProblemsMCMC_PhilipovGlickman}}
Following from Section~\ref{subsec:ModelPhilipovGlickman}, we now discuss why the MCMC in \citet{philipov+g06}, as corrected by \citet{rinnergschwenter+tw12}, is infeasible for this problem.
\citet{rinnergschwenter+tw12} shows that $A^{-1}$  cannot be 
sampled directly from a Wishart distribution in a Gibbs sampling step. We tried implementing a random-walk Metropolis-Hastings update for $A^{-1}$ using a Wishart proposal with a mean equal to the current value in the MCMC. The erroneous step also results in changes for all the full conditionals, which explains why our implementation does not achieve the sampling efficiency reported
in \citet{philipov+g06} when $k=5$.
Figure~\ref{fig:ACF_PhilipovGlickman} shows that the MCMC iterates are highly correlated,
leading to small values for the effective sample size;
 leads to unreliable highly variable kernel density estimates of posterior densities.
and is the reason we benchmark the adequacy of the variational approximation using the oracle predictive density approach in Section~\ref{subsec:Results_PhilipovGlickman}.

Therefore, this application is infeasible for this particular Metropolis-Hastings within Gibbs sampler when $k=5$; it gets even worse for $k=12$, and the main reasons the sampler fails are due to an independent Wishart proposal within Gibbs for updating $\Sigma^{-1}_t$ for $t<T$ (at $t=T$ perfect sampling from a Wishart can be applied) and the random-walk proposal within Gibbs for $A^{-1}$. It is well known that these proposals fail in a high-dimensional setting: the random-walk explores the sampling space very slowly while independent samplers get stuck, i.e. reject nearly all attempts to move the Markov chain.

We remark that other MCMC approaches for estimating this model more efficiently might be possible, but it is outside the scope of this paper to pursue this. As an example, Hamiltonian Monte Carlo on the Riemannian manifold \citep{girolami2011riemann} has proven to be effective in sampling models with $500\text{-}1,000$ parameters. However, the computational burden relative to standard MCMC is increased and, moreover, tuning the algorithm becomes more difficult.

\begin{figure}[h]
\centering
\includegraphics[width=0.6\columnwidth]{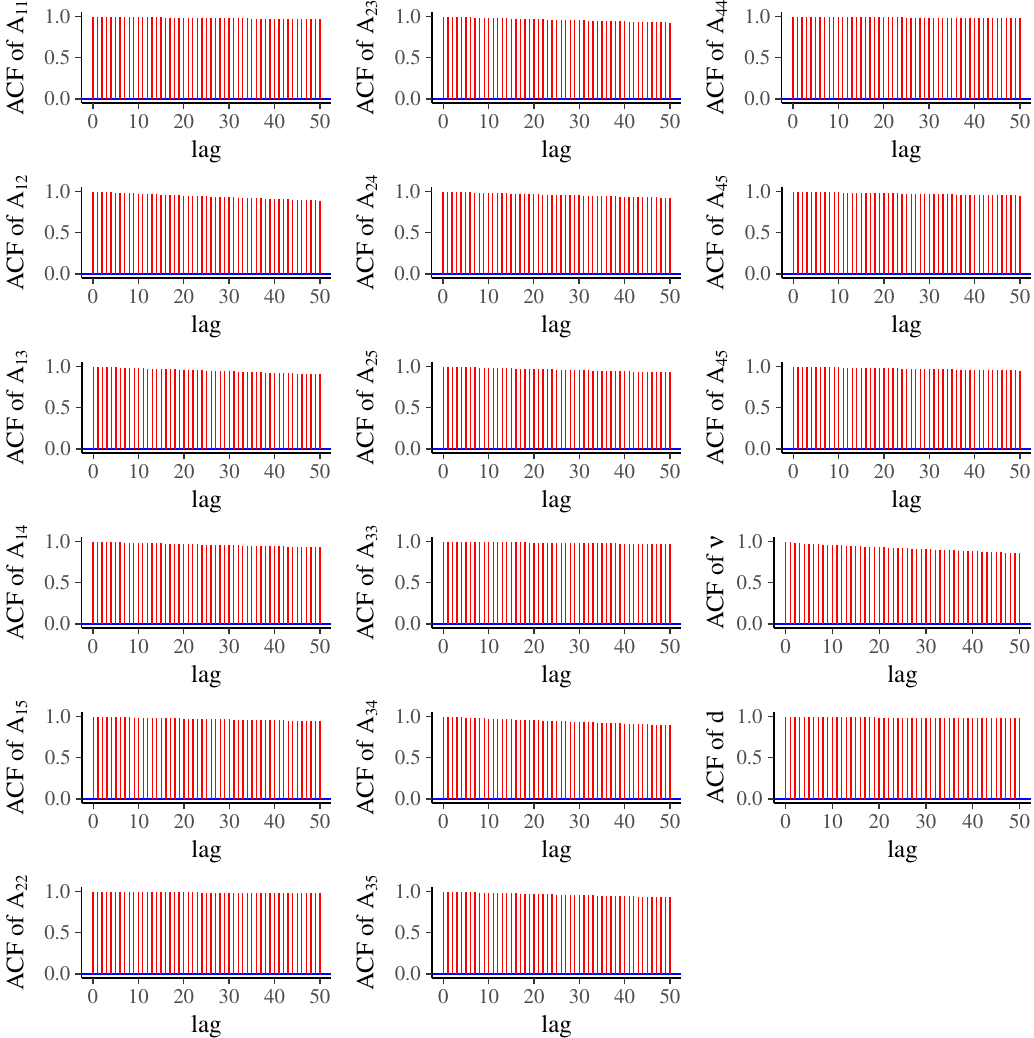}
\caption{Autocorrelation function (ACF) plots for the MCMC iterates of $A$, $\nu$ and $d$.}
\label{fig:ACF_PhilipovGlickman}
\end{figure}

%\newpage
%\bibliographystyle{apalike}
%\addcontentsline{toc}{section}{\refname}
%\bibliography{ref}
%\end{document}

\end{document}